\documentclass[12pt]{amsart}
\usepackage{amsfonts,amssymb,graphicx,mystyle,color, appendix}
\usepackage[colorlinks,linkcolor=blue]{hyperref} 
\usepackage{tikz}
\usepackage{cite}
\usepackage{hyperref}
 \let\backslash=\setminus \let\ge=\geqslant \let\le=\leqslant 
\def\a{\alpha} \def\b{\beta} \def\d{\delta} \def\e{\varepsilon}  \def\g{\gamma} \def\l{\lambda} \def\s{\sigma} \def\t{\tau} \def\z{\zeta}
\def\D{\Delta} \def\G{\Gamma} \def\L{\Lambda} \def\S{\Sigma}
  \def\F{\mathcal F}  \def\H{\mathcal H} \def\N{\mathcal N} \def\K{\mathcal K}

 \def\T{\mathcal T}  
\def\Re{\mathbb{R}}

\def\rho{\varrho}
\def\L{\Lambda}
\def\calG{\mathcal{G}}
\def\E{\mathcal E}

\def\I{\mathcal I}

\def\BR{\text{BR}}
\def\int{\text{int}}

\begin{document} \openup 1\jot

\title[On Sustainable Equilibria]{On Sustainable Equilibria\\}\thanks{A two-page abstract of a previous version was published in the Proceedings of the 21st ACM conference on Economics and Computation (EC 2020), July 13-17, 2020, Virtual Event, Hungary. ACM, New York, NY, USA, 2 pages. https://doi.org/10.1145/3391403.3399514}
\author[S. Govindan]{Srihari Govindan}
 \address{Department of Economics, University of Rochester, NY 14627, USA.}
 \email{s.govindan@rochester.edu}
 \author[R. Laraki]{Rida Laraki}
 \address{CNRS, Lamsade, University of Paris Dauphine-PSL, 75016 Paris, France, and Department 
 \indent of Computer Science, University of Liverpool, Liverpool  L69 3BX, UK}
 \email{rida.laraki@dauphine.fr}
 \author[L. Pahl]{Lucas Pahl}\thanks{We would like to thank Paulo Barelli, Josef Hofbauer, Andy McLennan, Klaus Ritzberger, Francesc Dilm\'e and Sylvain Sorin for their extensive and useful comments, as well as the participants of Bonn, Glasgow, Holloway and Rochester Economic Theory Seminars, EC'2020 and SMAI MODE 2020 conferences, and the One World Mathematical Game Theory Seminar.}
 \address{Institute for Microeconomics, University of Bonn, Adenauerallee 24-42, 53113 Bonn, Germany}
 \email{pahl.lucas@gmail.com}
 \date{First version: January, 2020; This version: 5 August 2021.}

\begin{abstract}
Following the ideas laid out in Myerson \cite{M1996}, Hofbauer \cite{H2000} defined a Nash equilibrium of a finite game as sustainable if it can be made the unique Nash equilibrium of a game obtained by deleting/adding a subset of the strategies that are inferior replies to it.  This paper proves two results about sustainable equilibria.  The first concerns the Hofbauer-Myerson conjecture about the relationship between the sustainability of an equilibrium and its index: for a generic class of games, an equilibrium is sustainable iff its index is $+1$. Von Schemde and von Stengel \cite{vSvS2008} proved this conjecture for bimatrix games;  we show that the conjecture is true for all finite games. More precisely, we prove that an isolated equilibrium has index $+1$ if and only if it can be made unique in a larger game obtained by adding finitely many strategies that are inferior replies to that equilibrium. Our second result gives an axiomatic extension of sustainability to all games and shows that only the Nash components with positive index can be sustainable.
\end{abstract}

\maketitle

\section{Introduction}

Myerson \cite{M1996} proposes a refinement of Nash equilibria of finite games, which he calls {\it sustainable equilibria}, based on the hypothesis that most games, even if they are one-shot affairs, should be analyzed not as if they are played in isolation, but rather as particular instances of many plays of such games.  Myerson argues that when, say,  two members of a society play a Battle-of-Sexes game, if the game has a history in this society, it becomes a ``culturally familiar'' game for these two players, and the past history of plays, by other members of the society, should inform play in this encounter.  An equilibrium is then a cultural norm, an institution, in this society, and the game is typically played according to this norm.  Any Nash equilibrium of the underlying game that can emerge as a norm in some society is sustainable. From this perspective, Myerson reasons,  the two pure-strategy equilibria in the Battle-of-Sexes game are sustainable while the mixed equilibrium is not. 

In his search for a formal definition of a sustainable equilibrium, Myerson considers, and then dismisses, on axiomatic grounds, existing refinements that yield the same prediction in the Battle-of-Sexes game as his heuristic argument does:  for example, persistent equilibria \cite{KS1984} fail invariance \cite{B1992};  and evolutionary stability fails existence. Myerson concludes his paper with a conjecture that the index of an equilibrium is a determinant of its sustainability.\footnote{Interestingly, Myerson speculates that one could perhaps develop a theory of index of equilibria based on fixed-point theory, seemingly unaware, as Hofbauer \cite{H2000} observes, of an extant theory in the literature (see G\"{u}l, Pearce and Stacchetti \cite{GPS1993} and Ritzberger \cite{R1994}).}

Hofbauer \cite{H2000} distills the ideas in Myerson's paper to provide a definition of sustainable equilibria in regular games\footnote{Call an equilibrium  {\it regular} if locally the equilibrium is a differentiable function of the payoffs (see section 2.1).}.  Hofbauer posits that a minimum requirement of sustainability should be that if an equilibrium of a game is sustainable, it should remain sustainable in the game obtained by restricting players' strategies to the set of best replies to the equilibrium.\footnote{For regular games, this is equivalent to restricting players' strategies to the support of the equilibrium---see Section 4 for a discussion of this point.} If one also accepts that an equilibrium that is unique is sustainable, then one is lead to the following definition. Say that a game-equilibrium pair is {\it equivalent} to another such pair if the restrictions of the two games to the set of best replies to their respective equilibria are the same game (modulo a relabelling of the players and their strategies) and the two equilibria coincide (under the same identification). An equilibrium of a game is {\it sustainable} if it has an equivalent pair where the equilibrium is unique.

In this paper, we prove two results about sustainable equilibria.  The first result is about sustainable equilibria of (generic) games where all equilibria are regular. Following Myerson, Hofbauer conjectured that a regular equilibrium is sustainable iff its index is $+1$.  Von Stengel and von Schemde \cite{vSvS2008} proved this conjecture for bimatrix games.  In this paper, we show that the conjecture holds for all $N$-person games.   

While the definition of sustainable equilibria does not assume that the game is regular, the existence of such equilibria is not guaranteed in nongeneric games.  Our second result provides a resolution to this problem via an axiomatic approach.  We enumerate a set of axioms that for regular games selects sustainable equilibria and for the other games selects components with positive index.

Our proof of the the Hofbauer-Myerson conjecture holds in all games with isolated equilibria. We prove that an isolated equilibrium $\sigma$ of a finite game $G$ has index $+1$ if and only if one can add finitely many new strategies together with their payoffs, all inferior replies to $\sigma$, and obtain a larger game $\hat G$ in which $\sigma$ is the unique equilibrium. To illustrate the statement, consider the battle of the sexes game $G$ below.  It has three Nash equilibria: two are strict $(t,l)$ and $(b, r)$ (with index $+1$), and one is mixed (with index $-1$).
$$G=
\begin{array}{ccc}
\multicolumn{1}{c}{} &\multicolumn{1}{c}{l} &\multicolumn{1}{c}{r} \\
\cline{2-3}
\multicolumn{1}{c}{t} &\multicolumn{1}{|c}{(3,2)} &
\multicolumn{1}{|c|}{(0,0)} \\
\cline{2-3}
\multicolumn{1}{c}{b} &\multicolumn{1}{|c}{(0,0)} &
\multicolumn{1}{|c|}{(2,3)}\\
\cline{2-3}
\end{array}
$$

By adding one strategy to each player ($x$ for player 1, and $y$ for player 2, see the game $\hat G$ below), $b$ and $r$ are now strictly dominated. Removing them yields a game where $x$ and $y$ are strictly dominated, making the strict equilibrium $(t,l)$ of $G$ the unique equilibrium of $\hat G$.
$$
\hat G=
\begin{array}{cccc}
\multicolumn{1}{c}{} &\multicolumn{1}{c}{l} &\multicolumn{1}{c}{r} &\multicolumn{1}{c}{y} \\
\cline{2-4}
\multicolumn{1}{c}{t} &\multicolumn{1}{|c}{(3,2)} &
\multicolumn{1}{|c|}{(0,0)} &
\multicolumn{1}{|c|}{(0,1)} \\
\cline{2-4}
\multicolumn{1}{c}{b} &\multicolumn{1}{|c}{(0,0)} &
\multicolumn{1}{|c|}{(2,3)}&
\multicolumn{1}{|c|}{(-2,4)}\\
\cline{2-4}
\multicolumn{1}{c}{x} &\multicolumn{1}{|c}{(1,0)} &
\multicolumn{1}{|c|}{(4,-2)}&
\multicolumn{1}{|c|}{(-1,-1)}\\
\cline{2-4}
\end{array}
$$

This construction easily extends to all finite games: any strict equilibrium (necessarily pure, regular and with index $+1$, see Ritzberger \cite{R1994}) can be made the unique equilibrium in a larger game obtained by adding finitely many strategies that are inferior replies to the equilibrium.\footnote{To make a strict equilibrium $s=(s_1,...,s_N)$ of a game $G$ unique in a larger game $\hat G$, it is enough to add one strategy $x_n$ per player $n$ as in the Battle-of-Sexes, where for each player $n$, $x_n$ strictly dominates all its pure strategies $t_n \neq s_n$; once the strategies $t_n \neq s_n$ are eliminated, $x_n$ becomes strictly dominated by $s_n$.} Our paper not only extends this property to isolated mixed equilibria with index $+1$ but shows that they are the unique equilibria having that property, that is, if a Nash equilibrium (isolated or not) can be made unique by adding finitely many inferior replies to it, that equilibrium must be isolated and must have index +1.

What are the key properties of the index of equilibria that drive this equivalence? To answer this question, let us see a sketch of our proof. In one direction, suppose that an equilibrium $\s$ of a game $G$ is sustainable, and that $(G, \s)$ is equivalent to a pair $(\bar G, \s)$ where $\s$ is the unique equilibrium of $\bar G$. Let $G^*$ be the game obtained from $G$ by deleting strategies that are inferior replies to $\s$. It follows from a property of the index that the index of $\s$ in $G$ can be computed as the index of $\s$ in $G^*$. The game $G^*$ is also the game obtained from $\bar G$ by deleting inferior replies there.  Therefore, the index of $\s$ in $G^*$ can also be computed as the index of $\s$ in $\bar G$.  As $\s$ is the unique equilibrium of $\bar G$, its index is $+1$, which then gives us the result.

Going the other way, if we have a $+1$ index equilibrium $\s$ of a game $G$, the sum of the indices of the other equilibria is zero, as the sum of indices over all equilibria is $+1$. Now, we can take a map whose fixed points are the Nash equilibria of $G$ and alter it outside a neighbourhood of $\s$ so that the new map has no fixed points other than $\s$.\footnote{The possibility of such a construction follows from a deep result in algebraic topology, the Hopf Extension Theorem (Corollary 8.1.18, \cite{S1966}).} By a careful addition of strategies and specification of payoffs for these strategies, we obtain a game $\bar G$ where any equilibrium must translate to a fixed point of the modified map of $G$, making $\s$ the unique equilibrium in $\bar G$.

A word about our methodology is in order.  In a bimatrix game, a player's payoff function is linear in his opponent's strategy and the index can be computed easily using the Shapley formula \cite{S1974}. Von Schemde and von Stengel \cite{vSvS2008} were able to exploit those features and use tools from the theory of polytopes (von Schemde  \cite{vS2005}) to prove the conjecture.  In the general case, their technique is inapplicable. What we do, instead, is start with a construction involving a fixed-point map and then convert it into a game-theoretically meaningful one.  In this respect, our approach is similar in spirit to, but different in details from, that in Govindan and Wilson \cite{GW2005}. 

Equilibria with index $+1$ are also distinguished from their counterparts with index $-1$ in terms of their dynamic stability.  It is well-known, both in general equilibrium and in game theory, that equilibria with index $-1$ are asymptotically unstable under any reasonable learning or adjustment process---cf. McLennan \cite{M2016}.\footnote{In fact, this observation leads McLennan to articulate his {\it index $+1$ principle}, which selects $+1$ equilibria.}   Even computational dynamics like those generated by homotopy algorithms (Lemke-Howson, linear-tracing procedure, etc.) converge to a $+1$ index equilibrium (Herings-Peeters \cite{HP2010}).  While these results might be seen as eliminating $-1$ equilibria, they do not come down conclusively in favor of all $+1$ equilibria.  The main reason is that it is still an open question as to whether every regular game has at least one equilibrium (necessarily of index $+1$) that is asymptotically stable with respect to some natural dynamical system. Hofbauer observed that some $+1$ index equilibria are indeed unstable for all natural dynamics, as the following potential game $G_1$ shows.\footnote{Note however that in $G_1$, the three pure equilibria are strict, have index $+1$ and can be shown to be asymptotically stable for all natural dynamics.}
$$
G_1=
\begin{array}{cccc}
\multicolumn{1}{c}{} &\multicolumn{1}{c}{l} &\multicolumn{1}{c}{m} &\multicolumn{1}{c}{r} \\
\cline{2-4}
\multicolumn{1}{c}{t} &\multicolumn{1}{|c}{(10,10)} &
\multicolumn{1}{|c|}{(0,0)} &
\multicolumn{1}{|c|}{(0,0)} \\
\cline{2-4}
\multicolumn{1}{c}{m} &\multicolumn{1}{|c}{(0,0)} &
\multicolumn{1}{|c|}{(10,10)}&
\multicolumn{1}{|c|}{(0,0)}\\
\cline{2-4}
\multicolumn{1}{c}{b} &\multicolumn{1}{|c}{(0,0)} &
\multicolumn{1}{|c|}{(0,0)}&
\multicolumn{1}{|c|}{(10,10)}\\
\cline{2-4}
\end{array}
$$

The profile where players mix uniformly is an isolated equilibrium with index $+1$ and so it can be made unique in a larger game, for example by adding the three strategies $x$, $y$, and $z$ as in $\hat G_1$ below.\footnote{We refer the interested reader to von Schemde  \cite{vS2005}, p. 89-91, to understand how the strategies $x$, $y$ and $z$ are geometrically constructed.} The fact is that all natural dynamics increase the potential of $G_1$; since the completely-mixed equilibrium minimizes that potential, it is unstable for all natural dynamics---cf. Hofbauer \cite{H2000} for more about his second conjecture, which states that any regular game has at least one $+1$ index equilibrium that is asymptotically stable w.r.t. some natural dynamics.\footnote{It follows from Demichelis and Ritzberger \cite{DR2003} that Hofbauer's second conjecture is false if we include all games: a necessary condition for a component of Nash equilibria to be asymptotically stable is that its index agree with its Euler characteristic; yet, there are examples (e.g. \cite{HH2002} Fig. 8 or \cite{R2002} p. 325) where all components of equilibria are convex (and so have Euler characteristic $+1$), but no component has index $+1$.} 

$$
\hat G_1=
\begin{array}{ccccccc}
\multicolumn{1}{c}{} &\multicolumn{1}{c}{l} &\multicolumn{1}{c}{m} &\multicolumn{1}{c}{r} &\multicolumn{1}{c}{x} &\multicolumn{1}{c}{y} &\multicolumn{1}{c}{z} \\
\cline{2-7}
\multicolumn{1}{c}{t} &\multicolumn{1}{|c}{(10,10)} &
\multicolumn{1}{|c|}{(0,0)} &
\multicolumn{1}{|c|}{(0,0)}&
\multicolumn{1}{|c|}{(0,11)}&
\multicolumn{1}{|c|}{(10,5)}&
\multicolumn{1}{|c|}{(0,-10)} \\
\cline{2-7}
\multicolumn{1}{c}{m} &\multicolumn{1}{|c}{(0,0)} &
\multicolumn{1}{|c|}{(10,10)}&
\multicolumn{1}{|c|}{(0,0)}
&
\multicolumn{1}{|c|}{(0,-10)}&
\multicolumn{1}{|c|}{(0,11)}&
\multicolumn{1}{|c|}{(10,5)}\\
\cline{2-7}
\multicolumn{1}{c}{b} &\multicolumn{1}{|c}{(0,0)} &
\multicolumn{1}{|c|}{(0,0)}&
\multicolumn{1}{|c|}{(10,10)}
&
\multicolumn{1}{|c|}{(10,5)}&
\multicolumn{1}{|c|}{(0,-10)}&
\multicolumn{1}{|c|}{(0,11)}\\
\cline{2-7}
\end{array}
$$

In Section 5 we tackle the question of extending sustainability to nongeneric games.  When a game has no isolated equilibrium, no equilibrium can be made unique by adding inferior replies to it.  Moreover, there are games where no subset of a connected equilibrium component can be made unique by adding inferior replies to it, and there are (non regular) games where all equilibria are isolated but none of them has index $+1$ and so none of the finitely many equilibria can be made unique by adding inferior replies to it. Our solution to guarantee existence in every games is to take an axiomatic view.  

First, the Hofbauer-Myerson conjecture provides an axiomatic characterization of sustainability for regular games.  Second, the solution concept that assigns to each game its components with positive index extends sustainability to the universal domain of finite games and it satisfies three other axioms, beyond those invoked for generic games: connectedness, invariance, and robustness. Finally, if we combine the last two axioms to obtain a strengthening of robustness, then we get the result that any extension of sustainability must select from among the components with a positive index.

The rest of the paper is organized as follows.  Section 2 sets up the problem and states our main theorem, which is about the Hofbauer-Myerson conjecture.  It also gives an informal summary of the theory of index of equilibria. Section 3 is devoted to proving the theorem. Section 4 provides a discussion of the role of unused best replies in the definition of sustainable equilibria. Section 5 extends sustainability to connected sets of equilibria and axiomatically characterises positive index Nash components in the same spirit as the uniform hyperstability characterisation of non-zero index components by Govindan and Wilson \cite{GW2005}. We have two appendices.  The first reviews a construct from the theory of triangulations; and the second provides the proof of a key lemma that is invoked in Section 5.

\section{Definitions and statement of the main theorem} A finite game in normal form is a triple $(\N, {(S_n)}_{n \in \N}, G)$ where:  $\N = \{\, 1, \ldots, N \, \}$  is the set of players, with $N \ge 2$; for each $n \in \N$, $S_n$ is a finite set of pure strategies; and, letting $S \equiv \prod_{n \in \N} S_n$ be the set of pure strategy profiles, $G: S \to \Re^{\N}$ is the payoff function. By a slight abuse of notation, we will refer to a game by its payoff function $G$.

Given a game $G$, for each $n$, let $\S_n$ be the set of $n$'s mixed strategies and let $\S \equiv \prod_{n \in \N} \S_n$. Also, for each $n$, $S_{-n} \equiv \prod_{m \neq n} S_n$, and $\S_{-n} \equiv \prod_{m \neq n} \S_n$. 
The payoff function $G$ extends to $\S$ in the usual way and we will denote this extension by $G$ as well.  

Define an equivalence relation on game-equilibrium pairs as follows.  For $i = 1,2$, let 
($(\N^i, {(S_n^i)}_{n \in \N^i}, G^i), \s^i)$ be a game-equilibrium pair, i.e., $\s^i$ is an equilibrium of $G^i$. Say  that $(G^1, \s^1) \sim (G^2, \s^2)$ if, up to a relabelling of players and strategies, the restriction of $G^1$ to the set of best replies to $\s^1$ is the same game as the restriction of $G^2$ to the set of best replies to $\s^2$, and the equilibria coincide under this identification.\footnote{Note that this definition requires that the number of players in an equivalence class be the same.}  It is easily checked that $\sim$ is an equivalence relation.

\begin{definition}\label{def sustainable}
An equilibrium $\s$ of a game $G$ is {\it sustainable} if $(G, \s) \sim (\bar G, \bar \s)$ for a game $\bar G$ where $\bar \s$ is the unique equilibrium. 
\end{definition}

Sustainability is a property of equivalence classes and we could say that the canonical representation of a sustainable equilibrium is the game-equilibrium pair where there are no inferior replies to the equilibrium.

\subsection{Index and degree of equilibria}\label{indexofequilibria}
Both the index and the degree of equilibria are measures of the robustness of equilibria to perturbations. They differ in the space of perturbations they consider (perturbations of  fixed-point maps vs payoffs perturbations) but ultimately agree with one another.\footnote{The reason we are defining and reviewing both concepts, when apparently one would do, is that they are both useful in the exposition.}  We start with the degree of equilibria.  For simplicity we give a definition of degree only for regular equilibria. This approach allows  to bypass the use of algebraic topology, but more importantly it is germane to our problem, as we are concerned only with regular equilibria in this paper. For the general definition, see for e.g., Govindan and Wilson \cite{GW2005}.

Fix both the  player set $\N$ and the strategy space $S$. The space of games with strategy space $S$ is then the Euclidean space $\G \equiv \mathbb{R}^{\N \times S}$ of all payoff functions $G$.  Let $\E$ be the graph of the Nash equilibrium correspondence over $\G$, i.e., $\E = \{\, (G, \s) \in \G \times \S \mid \s \, \text{is a Nash equilibrium of} \, G \, \}$. Let $proj: \E \to \G$ be the natural projection: $proj (G, \s) = G$. By the Kohlberg-Mertens Structure Theorem \cite{KM1986}, there exists a homeomorphism $h: \E \to \G$ such that $h^{-1}$ is differentiable almost everywhere.  Say that an equilibrium $\s$ of a game  $G$ is {\it regular} if $proj \circ h^{-1}$ is differentiable and has a nonsingular Jacobian at  $h(G, \s)$; and say that a game is {\it regular} if each equilibrium $\s$ of $G$ is regular. If an equilibrium $\s$ is regular, then it is a {\it quasi-strict equilibrium}---that is, all unused strategies are inferior replies\footnote{see Ritzberger \cite{R1994} and van Damme \cite{vD1987}.}---and locally, the equilibrium is a smooth, even analytic, function of the game; moreover, it is also a regular equilibrium in the space of games obtained by deleting the unused strategies or, indeed, by adding strategies that are inferior replies to the equilibrium. The set of games that are not regular is a closed subset of lower dimension---actually codimension one---in $\G$ and thus regular games are generic. 

If an equilibrium $\s$ of a game $G$ is regular, then we can assign a {\it degree} to it that is either $+1$ or $-1$  depending on whether the Jacobian of $proj \circ h^{-1}$ at $h(G, \s)$ has a positive or a negative determinant. An inspection of the formula for the Jacobian shows that the degree of a regular equilibrium $\s$ is the same as its degree computed in the space of games obtained by deleting the strategies that are inferior replies to $\s$. Therefore, if $\s$ is a regular equilibrium of $G$ and if $(G, \s) \sim (\bar G, \bar \s)$ then $\bar \s$ is a regular equilibrium of $\bar G$ and it has the same degree as $\s$, making degree an invariant for an equivalence class. 

As the Kohlberg-Mertens homeomorphism extends to the one-point compactification of $\E$ and $\G$, and $proj \circ h^{-1}$ is homotopic to the identity on this extension, the sum of the degrees of equilibria of a regular game is $+1$.\footnote{If $G$ is nongeneric, we can define the degree of a component of equilbria as the sum of the degrees of equilibria in a neighborhood of the component for a regular game that is in a neighborhood of $G$; this computation is independent of the neighborhoods chosen, as long as they are sufficiently small. The sum of the degrees of the components of equilibria of a game is $+1$.}

In fixed-point theory, the index of fixed points contains information about their robustness when the map is perturbed. (See McLennan \cite{M2018} for an account of index theory  written primarily for economists.) Since Nash equilibria are obtainable as fixed points, index theory applies directly to them. For simplicity, suppose $f: U \to \S$ is a differentiable map defined on a neighborhood $U$ of $\S$ in $\Re^{N|S|}$ and such that the fixed points of $f$ are the Nash equilibria of a game $G$. Let $d$ be the displacement of $f$, i.e., $d(\s) = \s - f(\s)$. Then the Nash equilibria of $G$ are the zeros of $d$.  Suppose now that the Jacobian of $d$ at a Nash equilibrium $\s$ of $G$ is nonsingular. Then we can define the {\it index} of $\s$ under $f$ as $\pm 1$ depending on whether the determinant of the Jacobian of $d$ is positive or negative. 

One potential problem with the definition of index is the dependence of the computation on the function $f$, as intuitively we would think of the index as depending only on the game $G$.  But, under some regularity assumptions on $f$, we can show that the index is independent of $f$. Specifically, consider the class of  continuous maps $F: \G \times \S \to \S$ with the property that  the fixed points of the restriction of $F$ to $\{\, G \, \} \times \S$ are the Nash equilibria of $G$. Demichelis and Germano \cite{DG2000} show that the index of equilibria is independent of the particular map in this class that is used to compute it; Govindan and Wilson \cite{GW1997} show that the degree is equivalent to the index computed using one of the maps in this class, the fixed-point map defined by G\"{u}l, Pearce and Stacchetti \cite{GPS1993}. Thus, the index and degree of equilibria coincide---see Demichelis and Germano \cite{DG2000} for an alternate, more direct, proof of this equivalence. Given these results, for a regular equilibrium, we can talk unambiguously of its index  and use the term degree  interchangeably with it.

\subsection{Games in strategic form}\label{subsec strategic form}
It is convenient for us  to work with a somewhat larger class of games than normal-form games, called strategic-form games, and in this subsection we will define these games---cf. Pahl \cite{LP2019} for an extensive treatment of these games.

A game in {\it strategic form} is a triple $(\N, {(P_n)}_{n \in \N}, V)$ where: $\N$ is the player set; for each $n$, $P_n$ is a polytope\footnote{A polytope is a convex hull of finitely many points.} of strategies;  $V: \prod_{n \in \N} P_n \to \Re^\N$ is a multilinear payoff function. Clearly any normal-form game is a strategic-form game. Going the other way,  given a strategic-form game $(\N, {(P_n)}_{n \in \N}, V)$, we can define a normal-form game $(\N, {(S_n)}_{n \in \N}, G)$ where for each $n$, $S_n$ is the set of vertices of $P_n$ and for each $s\in S=\prod_n S_n$, $G(s) = V(s)$; the polytope $P_n$ can be viewed as the quotient space of $\S_n$ obtained by identifying all mixed strategies that are duplicates of one another (i.e., induce the same payoffs for all players for any profile of strategies of $n$'s opponents).

\subsection{Statement of the main theorem}
The following theorem settles the Myerson-Hofbauer conjecture in the affirmative. 

\begin{theorem}\label{thm main}
A regular equilibrium is sustainable iff its index is $+1$. 
\end{theorem}

As the sum of the indices of the equilibria of a regular game is $+1$,  there is at least one with index $+1$.  Thus, we have the following corollary. 

\begin{corollary}\label{corollary main}
Every regular game has at least one sustainable equilibrium.  
\end{corollary}

\section{Proof of Theorem \ref{thm main}} We will present the proof in a sequence of steps, each of which will be carried out in a separate subsection. 

\subsection{The index of a sustainable equilibrium}
We begin with a proof of the necessity of the condition.  Let $\s^*$ be a regular equilibrium  of a game $G$ that is sustainable.   Let $(G, \s^*) \sim (\bar G, \bar \s)$, where $\bar \s$ is the unique equilibrium of $\bar G$.  As we saw in the previous section, the index is constant on an equivalence class.  Since $\bar \s$ is the unique equilibrium of $\bar G$, its index is $+1$, and the result follows.

\subsection{Preliminaries}\label{prelims}
The rest of the section is devoted to proving the sufficiency of the condition. In this subsection, we introduce some key ideas that we exploit in the proof.  

First, we gather a list of notational conventions to be used. Throughout Section 3 (but not in the Appendix) we use the $\ell_\infty$-norm on Euclidean spaces.  For any subset $A$ of a topological space $X$, we let $\partial_X A$ be its topological boundary and $\text{int}_X (A)$ its interior. If $C$ is a convex set in a Euclidean space, then $\partial C$ and  $\text{int} (C)$ refer to the boundary and the interior of $C$ in the affine space generated by $C$.  

\begin{definition}\label{bonusgame}Given a payoff function $G$, and a vector $h \in \prod_{n \in \N} \Re^{S_n}$, let $G \oplus h$ be the game where the payoff to player $n$ from a profile $s \in S$ is $G_n(s) + h_{n,s_n}$. \end{definition}

For a game $G$,  recall that Nash \cite{N1951} obtains its equilibria as fixed points of a map  on the strategy space. This function, which we denote by $f$, is defined as follows. For each $n, s_n$ and $\s$, let $\phi_{n,s_n}(\s) \equiv \max \, \{ \, 0, G_n(s_n, \s_{-n}) - G_n(\s) \, \}$ and $\phi_n \equiv \prod_{s_n \in S_n} \phi_{n, s_n}$; then 
\[
f_{n,s_n}(\s) \equiv \frac{ \s_{n,s_n} + \phi_{n, s_n}(\s)}{1 + \sum_{t_n \in S_n}\phi_{n, t_n}(\s)}.
\]
For each $n$, $f_n(\s) = \s_n$ iff $\phi_{n,s_n} = 0$ for each $s_n \in S_n$. 
If $f_n(\s) \neq \s_n$ for some $n$ and $\s$, then letting
\[
r_n(\s) = {\left(\sum_{t_n \in S_n}\phi_{n, t_n}(\s)\right)}^{-1}\phi_n(\s)
\]
and
\[
\l_n(\s) = \frac{1}{1+ \sum_{t_n \in S_n}\phi_{n, t_n}(\s)},
\]
we have 
\[
f_n(\s) = \l_n(\s) \s_n + (1-\l_n(\s))r_n(\s).
\]
Thus $f_n(\s)$ is an average of $\s$ and a mixed strategy $r_n(\s)$; $r_n(\s)$ has the following properties: (1) it assigns a positive probability to a pure strategy iff it does strictly better than $\s_n$---in particular, it assigns zero probability to some strategy in the support of $\s_n$, as $f_n(\s) \neq \s_n$; (2) it assigns the highest probabilities to the best replies to $\s$.

\begin{definition}A game $(\N, \bar S, \bar G)$ {\it embeds} $(\N, S, G)$ if: (1) for each $n$: $S_n \subseteq \bar S_n$; and (2) the restriction of $\bar G$ to $S$ equals $G$. \end{definition}

Again for notational convenience, we will talk of a game $\bar G$ embedding $G$.   When $\bar G$ embeds $G$, we view the set $\S$ of mixed strategies of $G$ as a subset of the set $\bar \S$ of mixed strategies in $\bar G$.  Obviously, if $\bar G$ embeds $G$ and $\s$ is an equilibrium of $\bar G$ where for each $n$, the strategies that are not in $S_n$ are inferior replies, then $(G, \s) \sim (\bar G, \s)$.  Our proof technique is to show that for each regular $+1$ equilibrium $\s^*$ of $G$, we can embed $G$ in a game $\bar G$ where $\s^*$ is the unique equilibrium and the newly added strategies are inferior replies to $\s^*$.

We say that a strategic-form game $\bar V$ embeds $G$ if the associated normal-form game $\bar G$, as defined in Subsection \ref{subsec strategic form}, embeds $G$, or equivalently for each $n$, each strategy in $S_n$ is a vertex of the polytope $\bar P_n$ of $n$'s strategies in $\bar V$, and $G(s) = \bar V(s)$ for all $s \in S$.  In our proof we construct embeddings of $G$ in strategic-form games $\bar V$ that have a simple structure: for each $n$, the strategies in $S_n$ span a face of $\bar P_n$.

\subsection{A Simple consequence of regularity}\label{subsection consequence} From now on fix a game $G$ and let $\s^*$ be a regular equilibrium with index $+1$. For each $n$, let $S_n^*$ be the support of $\s_n^*$. Our objective in this subsection is to record the following simple, and yet consequential, property of $\s^*$.  There exists $\bar \e  > 0$ such that: if $\s \neq \s^*$ is an equilibrium of $G$, then there exist two different players $n_1$, $n_2$,  such that for $i = 1, 2$, there exists $s_{n_i} \in S_{n_i}^*$ with  $\s_{n_i, s_{n_i}} < \s_{n_i, s_{n_i}}^* - \bar \e$.  Indeed if this property is not true, there exist a sequence $k \to \infty$, a  corresponding sequence $\s^k$ of equilibria converging to some $\s$, and a player $n$ such that: (1) $\s^k \neq \s^*$ for all $k$; and (2) for all $m \neq n$ and $s_{m} \in S_m^*$,  $\s^k_{m, s_m} \ge \s_{m, s_m}^* - k^{-1}$.  Therefore, $\s_m = \s_m^*$ for all $m \neq n$. But $\s_n \neq \s_n^*$ as $\s^*$ is regular and, hence, isolated. This implies that $\l\s + (1-\l)\s^*$ is an equilibrium for all $\l \in [0, 1]$, again contradicting the fact  that $\s^*$ is isolated. Thus, there exists $\bar \e$ with the stated property.\footnote{Note that this proof only uses the fact that $\sigma^*$ is isolated.}

For $0 < \e \le \bar \e$, and each $n$, let $B_n^\e$  be the set of $\s_n \in \S_n$ such that $\s_{n,s_n} \ge \s_{n,s_n}^* - \e$  for all $s_n \in S_n^*$; and let $B^\e$ be the set of $\s$ such that $\s_n$ is  not in $B_n^\e$ for at most one $n$. (N.B. $\prod_{n \in \N} B_n^\e \subsetneq B^\e$). Since $B^\e=\bigcup_n (\Sigma_n \times \prod_{m\neq n} B_m^\e)$ is the union of finitely many closed sets, it is closed.

\subsection{Killing all fixed points of $f$ other than $\s^*$}\label{killing}

From the viewpoint of fixed-point theory, our problem amounts to  embedding $\S$  as a proper face of a polytope $\bar \S$, extending $f$ (the Nash map) to a function $\bar f$ on it, and then modifying $\bar f$ such that its only fixed point is $\s^*$.  From a game-theoretic viewpoint, there is an additional problem introduced by the caveat that  $\bar f$ should, in a sense, be realizable as a fixed-point map of a game $\bar G$ that embeds $G$---i.e., a map whose fixed points are the equilibria of game $\bar G$.  In this subsection, we solve the first problem partially, by constructing a map $f^0$ that coincides with $f$ on  $B^\e$ for some $0 < \e \le \bar \e$ and that has no fixed points outside it. We will later use this map to construct the embedding $\bar G$.

If necessary by adding a strictly dominated strategy for each player, we can assume that $\s_n^*$ belongs to $\partial \S_n$ for each $n$. (Recall that sustainability and index are properties of equivalence classes of regular equilibria, so that the addition of such strategies is harmless.) Let $V \equiv B^{\bar \e}$; $X \equiv \S \backslash \int_{\S} (V)$. The boundary of $X$ is relative to the affine space generated by $\S$, i.e., $\partial X \equiv (\partial \S \backslash V) \cup \partial_\S V$.  We claim now that $(X, \partial X)$ is homeomorphic to a ball with boundary. Indeed, the desired homeomorphism can be constructed as follows. Pick a completely-mixed strategy-profile $\s^0$ such that $\s^0_{n,s_n} < \s_{n,s_n}^* - \bar \e$ for all $n$ and $s_n \in S_n^*$. (Such a choice is possible since $\s^*_n$ belongs to the boundary of $\S_n$ and if necessary, we can decrease the $\bar \e$ defined above.) The set $X$ is star-convex at $\s^0$: for each $\s \in X$, $\l \s + (1-\l) \s^0 \in X$ for all $\l \in [0, 1]$.  Therefore, there is now a closed ball around $\s^0$ in $\S \backslash \partial \S$ that is contained in $X  \backslash \partial X$ that is homeomorphic to $X$ using radial projections from $\s^0$.

Define $\tilde f: \S \to \S$ as follows.  First, using Urysohn's lemma, construct a continuous function $\a: \S \to [0, 1]$ that is zero on $V$ and positive everywhere else.  Then, letting $\t_n^0$ be the barycenter of $\S_n$ for each $n$, define $\tilde f(\s) = (1- \a(\s))f(\s) + \a(\s) \t^0$. The function $\tilde f$ equals $f$ on $V$ and therefore $\s^*$ is the unique fixed point of $\tilde f$ in $V$. The set $\S \backslash V$ is mapped by $\tilde f$ to $\S \backslash \partial \S$. Hence all the other fixed points of $\tilde f$ belong to $X \backslash \partial X$.

Let $\tilde d$ be the displacement of $\tilde f$: $\tilde d(\s) \equiv \s - \tilde f(\s)$.  For each $n$, let $A_n$ be the hyperplane in $\Re^{S_n}$ through the origin and with normal $(1, \ldots, 1)$, and let $A = \prod_n A_n$. The map $\tilde d$ maps $\S$ into $A$. As the index of $\s^*$ is $+1$, the sum of the indices of the other components of fixed points of $\tilde f$, which are contained in $X \backslash \partial X$, is zero.  Therefore, $\tilde d: (X, \partial X) \to (A, A - 0)$ has degree zero.  By the Hopf Extension Theorem (cf. Corollary 8.1.18, \cite{S1966}) there exists a map $\hat d^0$ from $X$ to $A - 0$ such that its restriction to $\partial X$ coincides with $\tilde d$.  Extend $\hat d^0$ to the whole of $\S$ by letting it be $\tilde d$ outside $X$, i.e., on $V\backslash \partial_\S V$. 

For each $\s \in \S$, there exists $\l \in (0, 1]$ such that $\s - \l \hat d^0(\s) \in \S$: indeed, this is obvious for $\s \in X \backslash \partial X$, since $\s$ belongs to the interior of $\S$; for $\s \in \partial X \cup V$, we can take $\l = 1$ as $\hat d^0 = \tilde d$.  Now for each $\s$, let $\l(\s)$ be the largest $\l \in [0, 1]$ such that $\s - \l(\s)\hat d^0(\s) \in \S$. Note that the map $\s \mapsto \l(\s)$ is continuous. Define $f^0: \S \to \S$ by: $f^0(\s) = \s - \l(\s)\hat d^0(\s)$. The map $f^0$ is continuous, coincides with $f$ on $V$, and has $\s^*$ as its unique fixed point.

\subsection{Example}\label{Part1} 
We now introduce a running example, where we can carry out our construction numerically, and which we hope will aid in the understanding of the proof. The example differs from the text in one somewhat irrelevant respect: we focus on symmetric strategies, as it reduces the dimension of the problem and allows us to perform a two-dimensional graphical analysis as well. 

The game we study is a two-player coordination game given below. 
 $$
\begin{array}{ccc}
\multicolumn{1}{c}{} &\multicolumn{1}{c}{L} &\multicolumn{1}{c}{R} \\
\cline{2-3}
\multicolumn{1}{c}{L} &\multicolumn{1}{|c}{(1,1)} &
\multicolumn{1}{|c|}{(0,0)} \\
\cline{2-3}
\multicolumn{1}{c}{R} &\multicolumn{1}{|c}{(0,0)} &
\multicolumn{1}{|c|}{(1,1)}\\
\cline{2-3}
\end{array}
$$

\medskip
Given our restriction to symmetric strategies, we will dispense with the subscript for players in the notation (here and throughout the paper when we work with this example).  Thus, a symmetric mixed-strategy profile is represented by one number, $x \in [0, 1]$, where $x$ is the probability of playing $L$. There are two pure strategy equilibria: $x = 1$ and $x = 0$, both of which have index $+1$; and there is a mixed equilibrium, $x = 1/2$, which has index $-1$.  The restriction of the Nash map $f$ to symmetric strategies allows us to represent it as a function from $[0, 1]$ to itself.   By computation, we obtain: 

\[
f(x) \equiv
  \begin{cases}
                                   \frac{x}{1-2x^2 + x}& \text{if $x \in [0,1/2]$} \\
                                   
\frac{x -2x^2 +3x -1}{1 -2x^2+3x-1} & \text{if $x \in (1/2,1]$}
  \end{cases}
\]

The graph of $f$, with its three fixed points corresponding to the three equilibria of the game, is illustrated below in Figure 1.  Let $V \equiv [2/3, 1]$. We can directly construct a map $f^0$ as in the previous section, whose only fixed point is $x = 1$ (see the graph of $f^0$ in green in Figure 1).\footnote{From Figure 1, one can observe that it is impossible to construct a map which coincides with $f$ in the neighbourhood of $x=1/2$ (the equilibrium with index -1) and which has $1/2$ as the unique fixed point.} 

\[
  f^0(x) \equiv
  \begin{cases}
                                   \frac{7}{8}& \text{if $x \in [0,2/3]$} \\
                                   
\frac{x -2x^2 +3x -1}{1 -2x^2+3x-1} & \text{if $x \in (2/3,1]$}
  \end{cases}
\]

\begin{figure}[h]
\caption{Graphs of $f$ (black) and $f^0$ (green)}
\medskip
\includegraphics[scale=0.5]{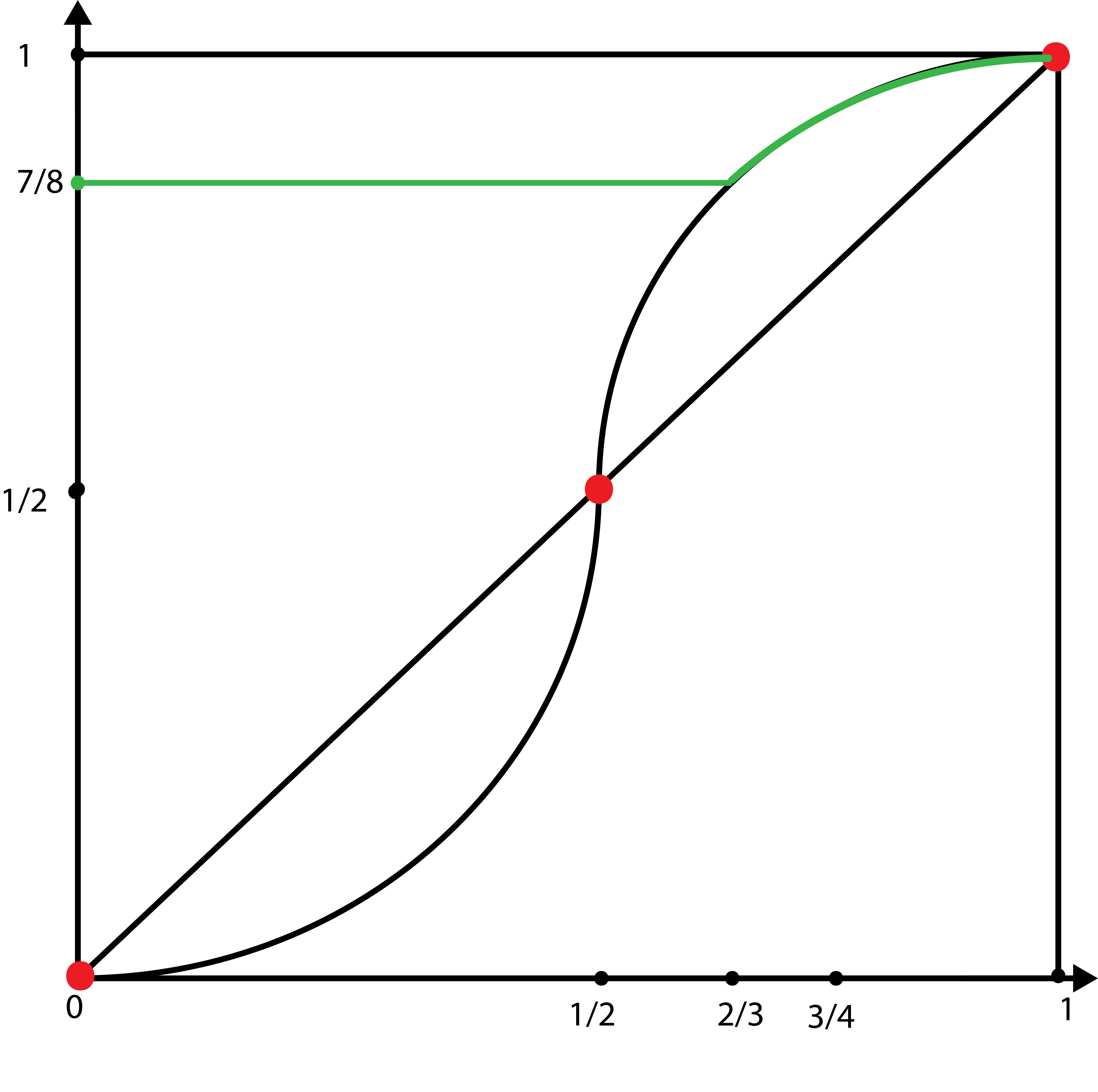}
\end{figure}

\subsection{A parametrized family of perturbed games}\label{parametrizedfamily}
Ideally, we would like a game $G^0$  such that $f^0$ is the Nash map of $G^0$. This seems to be too strong a property to hold.  However, $f^0$ does contain enough information for us to construct a function $g: \S \to \prod_{n \in\N} \Re^{S_n}$ such that: (1) $g(\cdot)$ is zero on $B^\e$ for some sufficiently small $\e$; (2) $\s$ is an equilibrium of $G \oplus g(\s)$ iff $\s = \s^*$. 

Choose $0 < \e < \bar \e $  and let  $U \equiv B^\e$; note that $U$ is a closed subset contained in the interior of $V$.   For each $n$, let $Z_n$ be ${(d_n^0)}^{-1}(0) \cap (\S \backslash \int_\S(U))$, where $d^0$ is the displacement of $f^0$.   Let $Z_n^1$ be the complement of $Z_n$ in $\S \backslash \int_\S(U)$. Define $r_n^0: Z_n^1 \to \partial \S_n$ as follows.  For each $\s \in Z_n^1$, let $r_n^0(\s)$ be the unique point in $\partial \S_n$ on the ray from $\s_n$ through $f_n^0(\s)$, i.e., it is the unique point of the form $(1 - \a) \s_n + \a f_n^0(\s)$ for $\a \ge 1$ that belongs to $\partial \S_n$.  If $\s \in Z_n^1$, then there exists some $t_n$ that is in the support of $\s_n$ but not of $r_n^0(\s)$. For each $n, s_n$, let $Z_{n,s_n}^+$ be the closure of the set of $\s \in Z_n^1$ for which $r_{n,s_n}^0(\s) \ge r_{n,t_n}^0(\s)$ for all $t_n \in S_n$.  If $f^0(\s) = f(\s)$ (the Nash map) and $\s \in Z_n^1$, then $r_n^0(\s)$ equals $r_n(\s)$ as defined in subsection \ref{prelims}; therefore, $\s \in Z_{n, s_n}^+$ iff $s_n$ is a best reply to $\s$ in $G$. 

We are now ready to define the function $g(\s)$. In doing so, we repeatedly invoke Urysohn's lemma to construct functions that are zero on a closed set and positive outside it.  First, let $v_n(\s) = \max_{s_n} G_n(s_n, \s_{-n})$.  Second, let $\beta_n^1: \S \to [0, 1]$ be a continuous function that is zero on $Z_n$ and positive everywhere else.  Third, for each $n, s_n$, let $\beta_{n,s_n}^2: \S \to [0, 1]$ be a continuous function that is one on $Z_{n,s_n}^+$ and  strictly smaller than one elsewhere.  Finally, let $\beta^3: \S \to [0, 1]$ be a continuous function that is one on $\S \backslash \int_\S(V)$, zero on $U$ and strictly positive everywhere else. For each $n, s_n$ and $\s$, define:
\[
g_{n,s_n}(\s) = \beta^3(\s)\beta_{n,s_n}^2(\s) [v_n(\s) - G_n(s_n, \s_{-n}) + \beta_n^1 (\s)].
\]
If $\s \in U$, then $g(\s) = 0$ as  $\beta^3(\s) = 0$; and $\s$ is an equilibrium of $G \oplus g(\s)$ iff $\s = \s^*$.  

Suppose $\s \notin U$. Since $\s^*$ is the only fixed point of $f^0$, there exists some $n$ such that $f_n^0(\s) \neq \s_n$. For this $n$, there exists $s_n$ such that: $\s \in Z_{n,s_n}^+$ (take $s_n$ s.t. $r^0_{n,s_n}(\s) \ge r^0_{n,t_n}(\s)$ for all $t_n \in S_n$);  and there is $t_n$ in the support of $\s_n$ but not in the support of $r^0_n(\s)$. This implies $\beta_{n,s_n}^2(\s) = 1$, while $\beta_{n,t_n}^2(\s) < 1$.  

If $\s \notin V$, then $\beta^3(\s) = 1$ and so, 
\[
G_n(s_n, \s_{-n}) + g_{n,s_n}(\s) = v_n(\s) + \beta_n^1(\s) > v_n(\s) + \beta_{n,t_n}^2\beta_n^1(\s) \ge G_n(t_n, \s_{-n}) + g_{n,t_n}(\s),
\]
showing that $\s$ is not an equilibrium of $G \oplus g(\s)$. 

If $\s \in V \backslash U$, then as $f^0$ coincides with $f$, $s_n$ is a best reply against $\s$ while $t_n$ is not.   Thus $G_n(s_n, \s_{-n}) = v_n(\s)$ and $G_n(t_n, \s_{-n}) <  v_n(\s)$. Since $\beta^3(\s)  > 0$, we obtain that 

\[
\begin{array}{lllll}
G_n(s_n, \s_{-n}) + g_{n,s_n}(\s) & = & v_n(\s) + \beta^3(\s)\beta_n^1(\s) & > & v_n(\s) + \beta^3(\s)\beta_{n,t_n}^2(\s)\beta_n^1(\s) \\
& & &  > & G(t_n, \s_{-n}) + g_{n,t_n}(\s),
\end{array}
\]
and again $\s$ is not an equilibrium of $G \oplus g(\s)$. Thus the function $g$ has the desired properties.

\subsection{Example}\label{Part2}

We continue with the example of subsection \ref{Part1}.  We will construct the function $g$ of the previous section.  (Recall our convention of dropping the player subscript for terms like $Z_n$, $Z_n^1$.)  What we are after is a function $g:[0, 1] \to \Re^{\{L, R\}}$ such that $x = 1$ is the only (symmetric) equilibrium of the game $G \oplus g(x)$.  Let $ U \equiv [3/4,1]$. First recall that $f^0$ has no fixed point outside $U$, so $Z$ is empty and $Z^1 = [0,3/4]$. For each $x \in [0,3/4]$, $f^0(x) > x$, which implies that $r^0_n(x) =1$ so that $Z^+_{L} = [0,3/4]$. Now, by computation one gets, 

\[
  v_n(x) \equiv
  \begin{cases}
                                   1-x & \text{if $x \in [0,1/2]$} \\
                                   
					x & \text{if $x \in (1/2,1]$}.
  \end{cases}
\]

Because $Z$ is empty, we can set $\b^1(\cdot)$ to be a constant function equal to $\d > 0$. The map $\b^3(\cdot)$ will be defined as follows:

\[
  \b^3(x) \equiv
  \begin{cases}
                                   1 & \text{if $x \in [0,2/3]$} \\
                                   -12x+9 & \text{if $x \in (2/3,3/4)$}\\
					0 & \text{if $x \in (3/4,1]$}.
  \end{cases}
\]

There is no need to introduce the function $\b^2(\cdot)$ in this example.  Putting these ingredients together, we can define $g$ as follows: $g_{L}(x) = \b^3(x) [v_1(x) - x + \d]$ and $g_{R}  \equiv 0$. We show that if the payoffs are now perturbed according to the bonus function $g$ for each player, the only remaining equilibrium is $x = 1$. Let $x$ be an equilibrium of the perturbed game. If $x \in [0, 2/3]$, $\b^3(x)=1$, so if player 1 plays $x$, player 2 gets $\d$ more than the best payoff $v_2(x)$ in the unperturbed game from playing $L$  whereas by playing $R$ he will not get more than $v_2(x)$. Therefore, $x =1$, which is a contradiction.  On the other hand, if $x \in [2/3, 1]$, then $L$ is already the strict best-reply in the original game and since the $g$ is nonnegative, it follows that $x = 1$ is the unique equilibrium of the perturbed game.

\subsection{Isolating  $\s^*$}
Before we can use the perturbation $g$, we need to first embed $G$ in a game $\tilde G$ where $\s^*$ is the only equilibrium in the face $\S$ of $\tilde G$ and in fact the only equilibrium in which the strategy of even one of the players is in $\S_n$.  (The perturbation $g$ is then used on the face opposite to $\S$.)  Hence, the embedding $\tilde{G}$ will be such that if $\tilde \s$ is an equilibrium of $\tilde G$ and the support of $\tilde \s_n$ is in $S_n$ for some $n$, then $\tilde \s = \s^*$.  The game  $\tilde G$ that embeds $G$ will be represented in strategic form.

Choose $0 < \e^* < \e$ such that $\s_{n,s_n}^* > \e^*$ for each $n$ and $s_n \in S_n^*$.  Let $U^*_n \equiv B^{\e^*}_n$ for each $n \in \N$, and let $U^* \equiv B^{\e^*}$. The set $U^*$ is a proper subset of $U$ (and is a closed subset contained in the interior of $V$).  For each $n$, choose an arbitrary object $0_n^*$ (not in $S_n^*$). Let $\Theta_n$ be the set of  distributions over $\prod_{m \neq n} (S_m^* \cup \{\, 0_m^* \, \})$.    
For each player $n$, his strategy set  $\tilde \S_n$ in the strategic form of $\tilde G$ is  $\S_n \times \Theta_n$.  A typical element $\tilde \s_n \in \tilde \S_n$ has coordinates $(\s_n, \theta_n)$. 

We will now describe the payoff functions. For each $\theta_n \in \Theta_n$ and $m \neq n$, we let $\theta_{n,m}$ be the marginal distribution of $\theta_n$ over $S_m^* \cup \{\, 0_m^* \, \}$ and let $\Theta_{n,m}$ be the set of all probability distributions over $S_m^* \cup \{\, 0_m^* \, \}$.  For each $m \neq n$, let $\g_{n,m}: \Theta_{n,m} \times \S_m \to \Re$ be a bilinear function defined as follows. For all $\s_m$, $\g_{n,m}(0_m^*, \s_m) = 0$, while for $s_m \in S_m^*$, $\g_{n,m}(s_m, \s_m) = 1 - {(\s_{m, s_m}^* - \e^*)}^{-1}\s_{m, s_m}$. For each $\tilde \s$, player $n$'s payoff in $\tilde G$ is:\footnote{Technically, $\Theta_{n,m}$ and $\g_{n,m}$ do not depend on $n$.}
\[
\tilde G_n(\tilde \s) = G_n(\s) + \g_n(\theta_n, \s_{-n}),
\]
where
\[
\g_n(\theta_n, \s_{-n}) \equiv \sum_{m \neq n} \g_{n,m}(\theta_{n,m}, \s_m).
\]

Notice that the payoff function of each player $n$ is affine over each strategy set $\tilde \S_m$, $m=1,...,N$, so $\tilde G$ is indeed a well-defined  game in strategic form. For each $n$, let $\theta_n^0 =  {(0_m^*)}_{m \neq n}$. Then $G$ is embeddable in $\tilde G$ as the face $\S_n \times  \{\, \theta_n^0 \, \}$ is a copy of the original face $\S_n$, for $n=1,...,N$.

Suppose $\tilde \s$ is an equilibrium of $\tilde G$, then $\s$ is an equilibrium of $G$, as the functions $\g_n$ of each player $n$ do not depend on $\s_n$.  If $\s = \s^*$, then the unique $\theta_{n,m}$ that is optimal for each $n \neq m$ is $0_m^*$ and thus the equilibrium uses $\theta_n^0$ for each $n$. On the other hand if $\s \neq \s^*$, by the property of subsection \ref{subsection consequence}, there are at least two players $m$ for whom $\s_m \notin U_m^*$.  Therefore, for each $n$, there is at least one $m \neq n$ such that $0_{n,m}^*$ is not optimal.  Thus for each $n$, the support of $\theta_{n}$ does not include $\theta_n^0$. 

To conclude, we showed that if $\tilde \s = {(\s_n, \theta_n)}_{n \in \N}$ is an equilibrium of $\tilde G$, then $\s$ is an equilibrium of $G$ and either: (1) $\s = \s^*$ and $\theta_n = \theta^0_n$, for each $n \in \N$; or (2) $\s \neq \s^*$ and the support of $\theta_n$ does not contain $\theta^0_n$ for any $n \in \N$. 

\subsection{Example}\label{Part3}
In the context of the example of subsection \ref{Part1}: $\Theta = \Delta (\{0^*\} \cup \{L\})$. We identify $\Theta$ with $[0,1]$ and an element $\theta_1 \in [0,1]$ denotes the probability of $L$. 
Let $\g(0^{*}, x) = 0$ and $\g(L, x) = 1 - \frac{8x}{7}$. Then define $\g(\theta, x) = \theta_1 \g(L, x) + (1-\theta_1)\g_1(0^*, x)$. The graph of $\g(L, x)$ is depicted below in red in Figure 2. Let $\e^* = 1/8$. 

\begin{figure}[h]
\caption{Graph of $\g(L, \cdot)$ (red)}
\medskip
\includegraphics[scale=0.5]{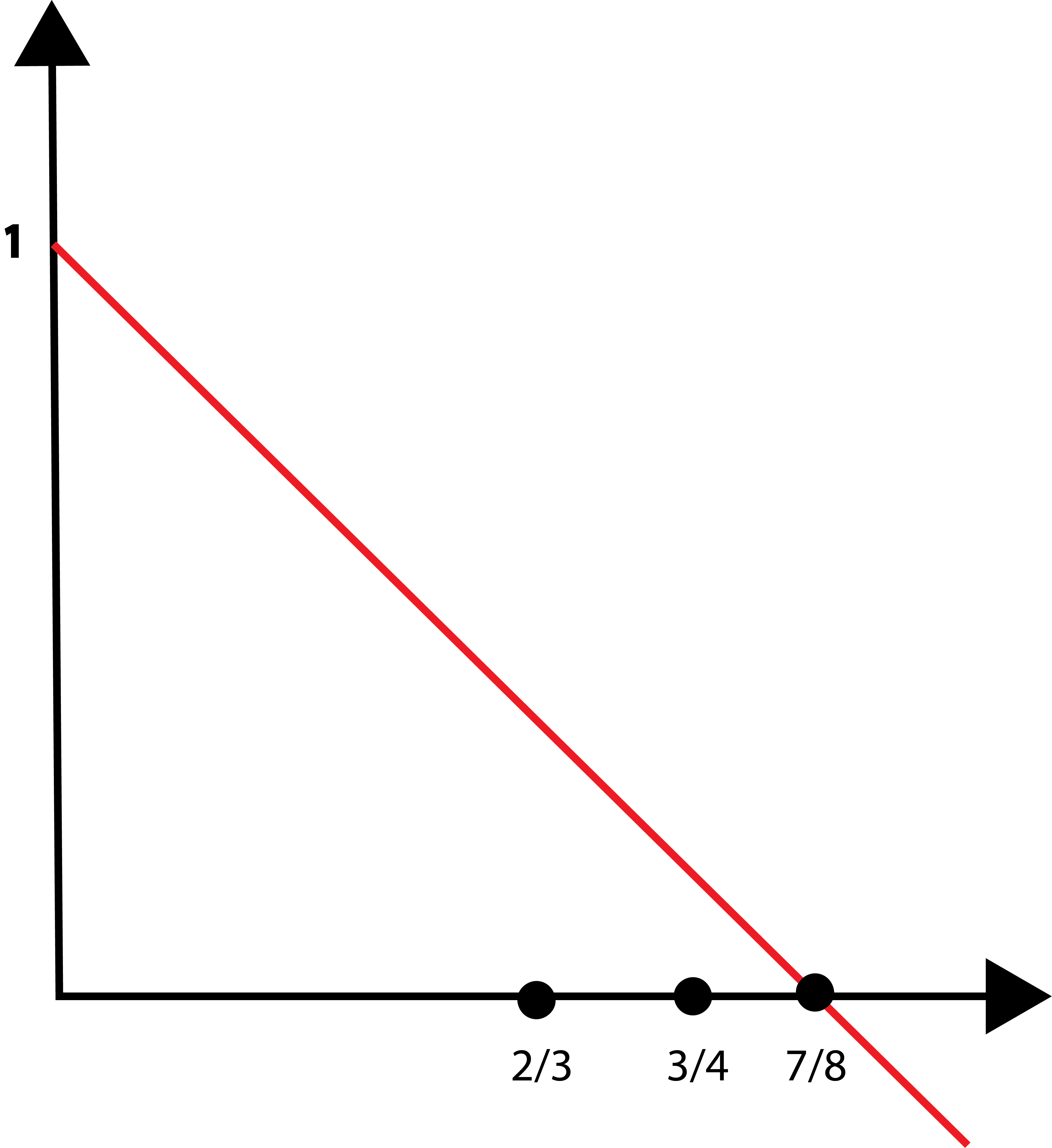}
\end{figure}

The strategic-form game is now defined by letting each player's strategy set be the square $[0,1] \times [0,1]$. The payoff function (which has to be defined for all profiles) is as follows. For player $1$, $\tilde G_1(x,\theta_1, y, \theta_2) = G_1(x,y) + \g(\theta_1, y)$; the payoffs for player 2 are defined symmetrically.  If player 2 plays $y < 7/8$, it follows that the (strict) best-reply of player 1 is to play $\theta_1 = 1$, in order to capture the positive bonus coming from $\g(L, x)$; if $y \ge 7/8$, then the bonus $\g(L,y)$ is nonpositive, and $\theta_1 = 0$ is a best-reply, which makes the payoffs of the perturbed game equal to the original payoffs. This game has a copy of the original equilibrium $x =1$: both players choose  $\theta = 0$ and $x=1$. Any other equilibrium is such that $\theta = 1$.

\subsection{The embedding $\bar G^\d$}\label{lastembedding} 
The embedding that allows us to obtain $\s^*$ as the unique equilibrium (and regular as well) will be built from $\tilde G$ by adding a finite number of mixed strategies as pure strategies and by defining their payoffs to eliminate all other equilibria.\footnote{von Schemde and von Stengel have an explicit bound on the number of strategies they add, which is three times the number of pure strategies in $G$.  In our construction, we have no way of obtaining such a bound: the construction depends on the fixed-point map $f^0$ obtained in subsection \ref{killing}, which in turn relies on an existence result, the Hopf Extension Theorem (cf. Corollary 8.1.18, \cite{S1966}).}  

The set $\S \backslash \int_\S(U^*)$ is compact, $g(\cdot)$ is continuous, and no $\s \in \S \backslash \int_\S (U^*) $ is an equilibrium of $G \oplus g(\s)$, as shown in subsection \ref{parametrizedfamily}. Hence, there exists $\eta > 0$ such that no $\s \in \S \backslash \int_\S (U^*)$ is an equilibrium of $G \oplus g$ for any $g$ with $\Vert g - g(\s) \Vert \le \eta$.  Also, since $g$ is uniformly continuous, there exists  $0 < \zeta < 1/3$ such that $\Vert g(\s) - g(\s') \Vert \le \eta$, if $\Vert \s - \s' \Vert \le \zeta$. Reduce $\zeta$ to ensure that it is also smaller than the distance between $U_n^*$ and $\partial_{\S_n} U_n$ for each $n$.

For each $n$, take a  triangulation $\T_n$ of $\S_n \times \Theta_n$  with the following properties.  (See the Appendix for the details.) (1) The only vertices in $\S_n \times \{\, \theta_n^0 \, \}$ of $\T_n$ are pure strategies $(s_n, \theta_n^0)$, $s_n \in S_n$; (2) letting $\Theta_n^1$ be the face of $\Theta_n$ where $\theta_n^0$ has zero probability, if $T_n\in \T_n$ is a simplex either with a face in $\S_n \times \Theta_n^1$, or shares a face with such a simplex,  then the diameter of $T_n$ is less than $\zeta$; (3) there exists a convex function $\rho_n: \S_n \times \Theta_n \to \Re_+$ such that: (a) $\rho_n(\l x + (1-\l) y) = \l \rho_n(x) + (1-\l) \rho_n(y)$ iff $x$ and $y$ belong to a simplex $T_n$ of $\T_n$; (b) $\rho_n^{-1}(0) = \S_n \times \{\, \theta_n^0 \, \}$. 

Let $\bar S_n^0$ be the set of vertices of the triangulation $\T_n$. Let $\bar S_{n}^1 \equiv \bar S_{n+1}^0$ modulo $N$.  A typical element of $\bar S_n^0$ is a pair $(\s_n, \theta_n)$ in $\S_n \times \Theta_n$ that is a vertex of $\T_n$. We fix a pure strategy $s_n^0 \in S_n$ for each $n$. These pure strategies will be used below to define the perturbation of payoffs $\pi_n$ for each player $n$. We denote a typical element of $\bar S_n^1$ by $(\s_{n, n+1}, \theta_{n, n+1})$, which is a vertex in $\T_{n+1}$. For $i = 0, 1$, let $\bar \S_n^i$ be the set of mixtures over $\bar S_n^i$. 
The pure strategy set of player $n$ in the game $\bar G^\d$ in normal form is $\bar S_n \equiv \bar S_n^0 \times \bar S_n^1$.  The set of mixed strategies is denoted $\bar \S_n$.  For each mixed strategy $\bar \s_n$,  and $i = 0, 1$, we let $\bar \s_n^i$  be the marginals over $\bar S_n^i$. Define $\bar S \equiv \prod_n \bar S_n$ and $\bar \S \equiv \prod_n \bar \S_n$. Also, let $\bar S^i \equiv \prod_n \bar S^i_n$ and $\bar \S^i \equiv \prod_n\bar\S_n^i$ for $i = 0, 1$.

Fix $\d > 0$. We will now define the  payoff function $\bar G^\d$.   For each $n$, let $\T_n^1$ be the collection of simplices of $\T_n$ that have nonempty intersection with $\S_n \times \Theta_n^1$. Given a pure strategy profile $\bar s \in \bar S$ with $\bar s_n = (\s_n, \theta_n, \s_{n,n+1}, \theta_{n, n+1})$ for each $n$, the payoff $\bar G^\d_n(\bar s)$ has five distinct components:
\[
\bar G^\d_n(\bar s) = G_n(\s) + \sum_{s_n \in S_n}  g_{n,s_n}^1(\bar s_{-n})\s_{n,s_n} + \g_n(\theta_n, \s_{-n}) + \pi_n(\bar s_n^1, \bar s_{n+1}^0) - \d \rho_n(\bar s_n^0). 
\]
The first and the third terms have been defined before. The function $\rho_n$ in the last term is the convex function defined above. We will specify the other two terms. 
\[
g_{n,s_n}^1(\bar s_{-n}) =  \xi_n(\bar s_{-n}^0)g_{n,s_n}(\s_1, \ldots, \s_{n-1},\s_{n-1, n}, \s_{n+1}, \ldots, \s_{N}),
\]
where $\xi_n(\bar s_{-n}^0)$ is one if for each $m \neq n$, $\bar s_m^0$ is a vertex of some simplex in $\T_m^1$; otherwise it is zero. The function $\pi_n$ is $0$ if either: (1) $\bar s_{n+1}^0$ is a vertex of some simplex in $\T_{n+1}^1$ and $\bar s_n^1 = \bar s_{n+1}^0$; or (2) $\bar s_{n+1}^0$ is not a vertex of such a simplex,  but $\bar s_n^1 = (s_{n+1}^0, \theta_{n+1}^0)$ ($s_{n+1}^0$ is a fixed pure strategy chosen above, while $\theta_{n+1}^0$ is the collection $(0_{n,m}^*)$); elsewhere it is $-1$. The definition of $\bar G^{\d}$ clearly implies that it embeds $G$.

We want to make a couple of remarks about the payoffs. First, the function $\pi_n$ incentivizes player $n$ to mimic player $n+1$ whenever the latter is choosing a strategy close to $\S_{n+1} \times \Theta_{n+1}^1$: if $n+1$ randomizes over the vertices of a simplex  $T_{n+1} \in \T^1_{n+1}$, then player $n$'s best replies must be among the vertices of the simplex. This will be a crucial property, since the choices in $\S_{n-1,n}$ play a role in the evaluation of the bonus function $g^1_n$. The idea is that whenever the bonus function $g^1_n$ is active, meaning that all players $m \neq n$ randomize over the vertices of a simplex $T_m$ in $\T^1_m$, then each pure best-reply for player $n$ must choose a vertex of $T_{n+1}$.  On the other hand, if player $n+1$ is randomizing over $S_{n+1} \times \{\, \theta_{n+1}^0 \, \}$ then it follows that the unique  best-reply for player $n$ is to choose the previously fixed strategies $\bar s^1_{n} = (s^0_{n+1}, \theta^0_{n+1})$.  

Our second remark concerns the nature of the payoffs for mixed strategies. For $n$ and each $i = 0, 1$, there is a linear map $p_n^i: \bar \S_n \to \S_{n+i} \times \Theta_{n+i}$ that sends each pure $\bar s_n$ to the corresponding mixed strategy in $\S_{n+i} \times \Theta_{n+i}$.  For each $n$, the first and the third terms of the payoffs depend on  $\bar \s$ only through their images under $p^0 = \prod_{n \in \N}p^0_n$; the fourth term depends on all the information in $\bar \s_n^1$ and $\bar \s_{n+1}^0$. The second term depends on $\bar \s_n$ only through $p_n^0$, but requires the entire information in $\bar \s_{-n}$, while the last term requires the information in $\bar \s_n^0$.

\subsection{Wrapping up the proof}
Let $\bar \s^*$ be the profile where for each player $n$, the marginal  on $\bar \S_n^0$ is $(\s^*, \theta_n^0)$ and the marginal on $\bar \S_n^1$ is $(s_{n+1}^0, \theta_{n+1}^0)$. For  $\d \ge 0$,  $\bar \s^*$ is an  equilibrium of $\bar G^\d$, and $(G, \s^*) \sim (\bar G^\d, \bar \s^*)$.   We will now show that for $\d$ sufficiently small, this is the only equilibrium of $\bar G^{\d}$, which completes the proof.

Say that a strategy  $\bar \s_n$ is \textit{admissible} for player $n$ if the support of its marginal $\bar \s_n^0 \in \bar \S_n^0$ is the set of vertices of a simplex $T_n$ in $\T_n$. Observe that for any $\d > 0$, every best reply for player $n$ is admissible. Indeed, the first three components of $n$'s payoff function depend on $n$'s strategy only through its projection to $\S_n \times \Theta_n$ and the fourth is independent of these choices. Therefore, any two strategies for $n$ that project under $p^0_n$ to the same point in $\S_n \times \Theta_n$ yield the same payoffs for these four terms, leaving the fifth to decide which one is better. But the map $\rho_n$ is convex, and it is linear precisely on the simplices  of $\T_n$, which then forces each best reply to be a mixture  over the vertices of a simplex of $\T_n$. 
 
We claim that if $\d = 0$ and the only admissible equilibrium of $\bar G^\d$ is $\bar \s^*$, then for sufficiently small $\d>0$, $\bar \s^*$ is the only equilibrium of $\bar G^\d$. To prove this claim,  suppose that we have a sequence $\bar \s^\d$ of equilibria of $\bar G^\d$ converging to some equilibrium $\bar \s^0$ of $\bar G^0$, then as we saw above $\bar \s^\d$ must be admissible, and hence also its limit $\bar \s^0$. As we have assumed that $\bar \s^*$ is the unique admissible equilibrium of $\bar G^0$, $\bar \s^0 = \bar \s^*$. Observe now that for each $n$, every pure best reply in $\bar G^0$ to $\bar \s^*$ is of the form $(s_n, \theta_n^0, s_{n+1}^0, \theta_{n+1}^0)$ where $s_n \in S_n$ is a best reply to $\s^*$; and this property holds for best replies to $\bar \s^\d$, for small $\d$. Thus  for each such $\d$, and for each $n$, $\bar \s_n^\d$ is of the form $(\s_n^{\d}, \theta_n^0, s_{n+1}^0, \theta_{n+1}^0)$, where $\s_n^\d$ is a best reply to $\s^*$ in $G$. In other words, $\s^\d $ is an equilibrium of $G$. As  $\s^\d$ converges to $\s^*$ and as $\s^*$ is an isolated equilibrium of $G$, $\s^\d = \s^*$ for all small $\d$. Thus the claim follows and it is  sufficient to show that $\bar \s^*$ is the only admissible equilibrium for $\d = 0$. 

To prove this last point, fix now an admissible equilibrium $\bar \s$ with marginals $(\bar \s^0, \bar \s^1) \in \bar \S^0 \times \bar \S^1$ of the game $\bar G^0$. For each $n$, let $(\s_n, \theta_n)$ and $(\s_{n,n+1}, \theta_{n, n+1})$ be the image of $\bar \s_n$ under $p_n^0$ and $p_n^1$, resp.  Also, let $T_n$ be the simplex of $\T_n$ generated by the support of $\bar \s_n^0$ for each $n$.

Suppose first for each $n$, $T_n$ belongs to $\T_n^1$. For each $n$, $\theta_n$ assigns probability less than $\zeta$, which is smaller than one, to $\theta_n^0$. Also, for at least two  $n$, $\s_n \notin \int_{\S_n}(U_n^*)$: indeed, otherwise there is one player $n$ all of whose opponents $m$ are choosing in $\int_{\S_n}(U_m^*)$, making $\theta_n^0$ the unique optimal choice, which is impossible. Thus, $\s_n \notin \int_{\S_n}(U_n^*)$ for at least two $n$, i.e., $\s \notin \int_{\S}(U^*)$ and, hence, $\s$ is not an equilibrium of $G \oplus g(\s)$.   For each $n$, and each $\bar s_{-n}$ in the support of $\bar \s_{-n}$, $\xi_n(\bar s_{-n}^0) = 1$ as $\bar s_m^0$ is a vertex of the simplex $T_m$, which is in $\T_m^1$, for each $m$; because of the function $\pi_n$, the optimality of $\bar \s_{n-1}^1$ implies that each $\bar s_{n-1}^1$ in the support of $\bar \s_{n-1}^1$ is a vertex of $T_n$.  Therefore, for each $\bar s_{-n}$ in the support of $\bar \s_{-n}$, $\Vert g^1(\bar s_{-n}) - g(\s) \Vert \le \eta$ and then $\Vert g^1(\bar \s_{-n}) - g(\s) \Vert \le \eta$. As $\s$ is not an equilibrium of $G \oplus g(\s)$, by the choice of $\eta$ in subsection \ref{lastembedding}, it is not an equilibrium of $G \oplus g^1(\bar \s)$, which contradicts the fact that $\bar \s$ is an equilibrium of $\bar G^0$. 

Now suppose that for exactly one $n$, say $n = 1$, $T_n$ does not belong to $\T_n^1$.  Then, $\theta_1^0$ has positive probability under $\theta_1$. Therefore, because of the definition of $\g_n$,  $\s_n \in U_n^*$ for $n > 1$, i.e., $\s \in U^*$. For $n > 1$, the fact that $\s_n \in U_n^*$ and $T_n$ belongs to $\T_n^1$  imply that for each $\bar s_n^0 = (\s_n, \theta_n)$ in the support of $\bar \s_n^0$, $\s_n$ belongs to $U_n$ (as the diameter of $T_n$ is less than $\zeta$, which is smaller than the distance between $U^*$ and $\partial_\S U$). Thus, $g_1^1(\s) = 0$. We will now show that $g_n^1(\s) = 0$ for $n > 1$. The payoff function $\pi_n$ for each $n \neq N$ forces each strategy $\bar s_n^1 = (\s_{n,n+1}, \theta_{n,n+1})$ in the support of $\bar \s_n^1$ to be a vertex of  $T_{n+1}$ and hence $\s_{n,n+1}$ is in $U_{n+1}$. Therefore, for $n>1$, $\s_{n-1,n} \in U_n$. Recall that $g_n(\cdot)$ was constructed to be $0$ on $U$. Consequently, for each $n > 1$, $g^1_n(\bar s_{-n}) = 0$ for each $\bar s_{-n}$ in the support of $\bar \s_{-n}$, i.e., $g^1_n(\bar \s_{-n}) = 0$. 

The fact that $g^1(\s) = 0$, implies that  $\s = \s^*$.  Optimality of $\theta_n$ for $n > 1$ now requires that it assign probability one to $\theta_n^0$. This is a contradiction: since $T_n \in \T^1_n$, its diameter is smaller than $\zeta$ (and hence one), putting it at positive distance from $\S_n \times \{\, \theta_n^0 \, \}$.

Finally, suppose that for at least two players $n$, $T_n$ does not belong to $\T_n^1$.  Then, again because of $\g_n$, for each $n$, $\s_n \in U^*_n$. We claim that for each $n$, $g_n^1(\bar s_{-n})= 0$ for each $\bar s_{-n}$ in the support of $\bar \s_{-n}$.  Indeed, if for some $m \neq n$, $T_m$ has no vertex in $\T_m^1$,  then $\xi_n$ is zero by construction at each $\bar s_{-n}$ in the support  and we are done.  Otherwise, if for each $m \neq n$, $T_m$ has a vertex in $\T^1_m$ then, letting $\bar s^0_m = (\s_m, \theta_m )$ be an arbitrary vertex of $T_m$, it follows from the fact that the diameter of each $T_m$ is less than $\zeta$ that $\s_m \in U_m$, which implies $\s \in U$. Therefore, $g_n^1(\cdot)$ is again zero on the support of $\bar \s_{-n}$.

It follows from the previous paragraph that $\s$ is an equilibrium of $G$, i.e., $\s  = \s^*$, making $\theta_n = \theta_n^0$.  Finally, optimality of $\bar \s_n^1$ implies that it is $(s_{n, n+1}^0, \theta_{n, n+1}^0)$, as it yields zero with others yielding $-1$.  Thus,  $\bar \s = \bar \s^*$, which concludes the proof.

\subsection{Example}\label{Part4}

We can triangulate the strategy set $\Sigma \times \Theta \equiv [0,1] \times [0,1]$ of each player as in Figure 3.  The horizontal axis represents $\Theta = [0,1]$ and the vertical axis $\Sigma = [0,1]$.

\begin{figure}[h]
\caption{Triangulation of $[0,1]^2$}
\medskip
\includegraphics[scale=0.7]{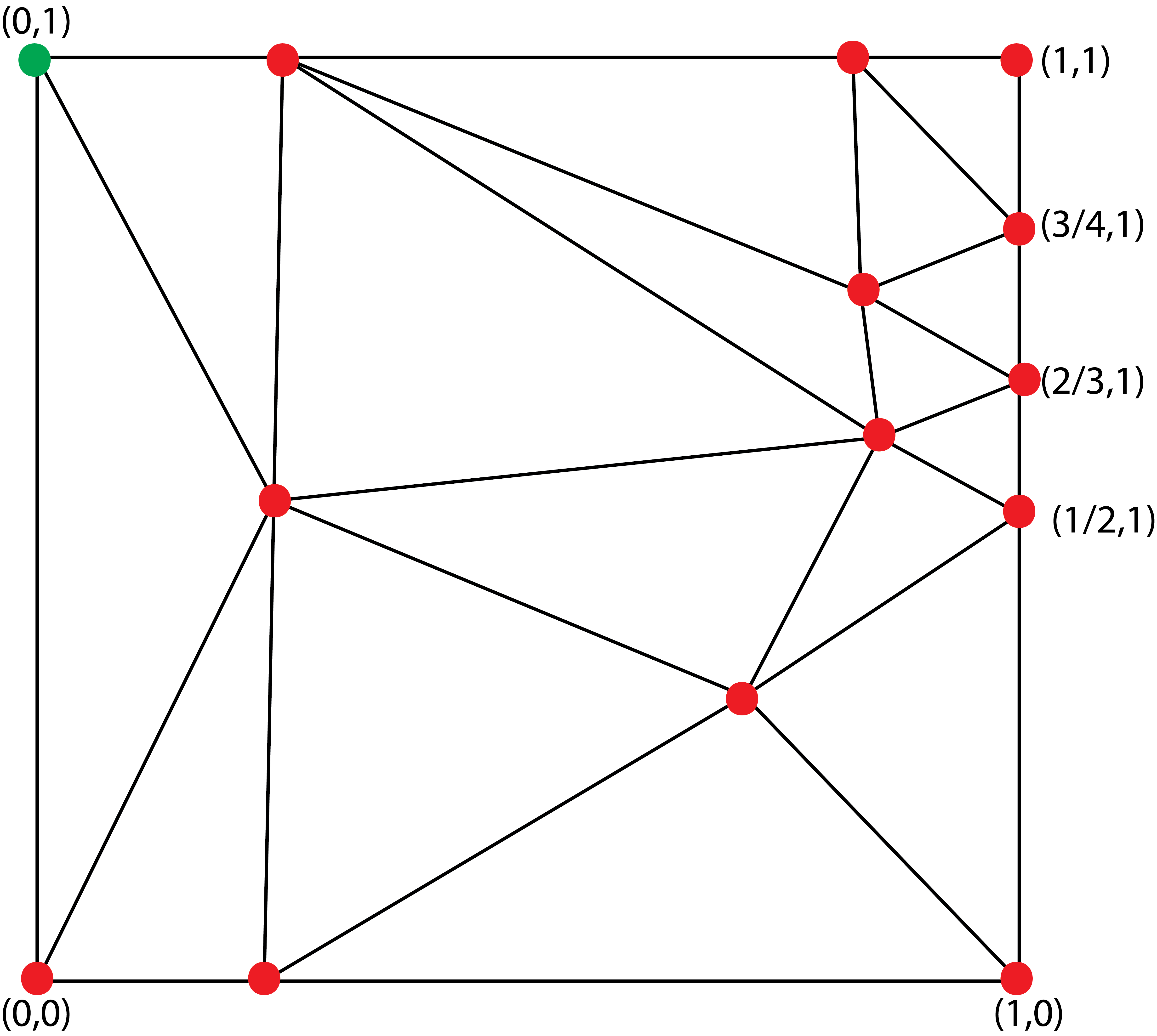}
\end{figure}

Payoffs are defined in the exact same way as in subsection \ref{lastembedding} with the following modifications: the function $\pi \equiv 0$, since there is no need to duplicate the strategy set of each player (by our construction, $g_n$ depends only on $\s_{-n}$); the function $\xi$ is equal to $1$ at all vertices of the triangulation that lie on the face $\theta = 1$.

Paralleling the proof of subsection \ref{lastembedding}, we show that the only admissible equilibrium of the perturbed game $\bar G^{\delta}$ with $\d =0$ is the (symmetric) equilibrium $(\theta, x)$, where $x = 1$ and $\theta = 0$. To see this, let $(x, \theta)$ be an equilibrium. Suppose first $x < 7/8$. Then $\theta =1$ is a strict best-reply. The support of the equilibrium $(\theta, x)$ is then a subset of one of the 1-dimensional simplices that subdivide $\{1\} \times [0,1]$. By our construction of $g$ and the fact that it is linear in each of these simplices, it follows that $x =1$, which is a contradiction. Therefore, $x \ge 7/8$. Recall that in the original game $x=1$ is a strict best-reply if $x \ge 7/8$. Since $g_{L} \ge 0$ (whereas $g_{R} \equiv 0$), it follows that $x=1$ is a strict-best reply in $\bar G^{\delta}$.  Finally, given $x = 1$, $\theta = 0$ is the optimal choice in $\Theta$.


\section{Deleting Unused Best Replies}
In the definition of equivalence between game-equilibrium pairs $(G, \s)$ and $(\bar G, \bar \s)$, we insisted that $G$ and $\bar G$ be the same game once we delete all strategies that are inferior replies to $\s$ and $\bar \s$, resp.  We could have weakened the requirement by allowing the deletion of unused best replies as well, i.e., that the games be the same once we restrict them the support of their respective equilibria.  For generic games, of course, these two notions coincide.
But, for nongeneric games, as we show now using examples, allowing for the deletion of unused best replies leads to an unsatisfactory concept of equivalence. 

Consider first the following bimatrix game:
$$
\begin{array}{ccc}
\multicolumn{1}{c}{} &\multicolumn{1}{c}{l} &\multicolumn{1}{c}{r} \\
\cline{2-3}
\multicolumn{1}{c}{t} &\multicolumn{1}{|c}{(1,1)} &
\multicolumn{1}{|c|}{(0,0)} \\
\cline{2-3}
\multicolumn{1}{c}{b} &\multicolumn{1}{|c}{(0,0)} &
\multicolumn{1}{|c|}{(0,0)}\\
\cline{2-3}
\end{array}
$$
The equilibrium $(b, r)$ is in dominated strategies and, clearly, an unreasonable equilibrium.  Deleting the strategies $t$ and $l$, which are ununsed best replies against this equilibrium, produces a trivial game where $(b,r)$ is now the only solution.

One might conjecture that the misbehavior in the above game stems from the fact that the equilibrium $(b, r)$ has index zero.  But that is not the case, as the following 3-player game  shows.  Player 1's strategy set is $\{\, (T,t), (T,b), B \, \}$; 2's strategy set is $\{\, (L,l), (L, r), R \, \}$; 3's strategy set is $\{ \, W, E^w, E^\e \, \}$. The payoffs are:

$$W:\;  \begin{array}{ccc}
\multicolumn{1}{c}{} &\multicolumn{1}{c}{L} &\multicolumn{1}{c}{R} \\
\cline{2-3}
\multicolumn{1}{c}{T} &\multicolumn{1}{|c}{
	\begin{array}{ccc}
	\multicolumn{1}{c}{} &\multicolumn{1}{c}{l} &\multicolumn{1}{c}{r} \\
	\cline{2-3}
	\multicolumn{1}{c}{t} &\multicolumn{1}{|c}{(6,6,1)} &
	\multicolumn{1}{|c|}{(0,0,1)} \\
	\cline{2-3}
	\multicolumn{1}{c}{b} &\multicolumn{1}{|c}{(0,0,1)} &
	\multicolumn{1}{|c|}{(6,6,1)}\\
	\cline{2-3}
	\end{array}
} &
\multicolumn{1}{|c|}{(3,3,0)} \\
\cline{2-3}
\multicolumn{1}{c}{B} &\multicolumn{1}{|c}{(3,0,1)} &
\multicolumn{1}{|c|}{(0,3,1)}\\
\cline{2-3}
\end{array}
$$

$$
E^w:\; \begin{array}{ccc}
\multicolumn{1}{c}{} &\multicolumn{1}{c}{L} &\multicolumn{1}{c}{R} \\
\cline{2-3}
\multicolumn{1}{c}{T} &\multicolumn{1}{|c}{
	\begin{array}{ccc}
	\multicolumn{1}{c}{} &\multicolumn{1}{c}{l} &\multicolumn{1}{c}{r} \\
	\cline{2-3}
	\multicolumn{1}{c}{t} &\multicolumn{1}{|c}{(-3,0,4)} &
	\multicolumn{1}{|c|}{(1,4,0)} \\
	\cline{2-3}
	\multicolumn{1}{c}{b} &\multicolumn{1}{|c}{(1,4,0)} &
	\multicolumn{1}{|c|}{(1,4,0)}\\
	\cline{2-3}
	\end{array}
} &
\multicolumn{1}{|c|}{(1,0,1)} \\
\cline{2-3}
\multicolumn{1}{c}{B} &\multicolumn{1}{|c}{(3,0,0)} &
\multicolumn{1}{|c|}{(0,3,0)}\\
\cline{2-3}
\end{array}\qquad
E^e:\; \begin{array}{ccc}
\multicolumn{1}{c}{} &\multicolumn{1}{c}{L} &\multicolumn{1}{c}{R} \\
\cline{2-3}
\multicolumn{1}{c}{T} &\multicolumn{1}{|c}{\begin{array}{ccc}
	\multicolumn{1}{c}{} &\multicolumn{1}{c}{l} &\multicolumn{1}{c}{r} \\
	\cline{2-3}
	\multicolumn{1}{c}{t} &\multicolumn{1}{|c}{(1,4,0)} &
	\multicolumn{1}{|c|}{(1,4,0)} \\
	\cline{2-3}
	\multicolumn{1}{c}{b} &\multicolumn{1}{|c}{(1,4,0)} &
	\multicolumn{1}{|c|}{(-3,0,4)}\\
	\cline{2-3}
	\end{array}
} &
\multicolumn{1}{|c|}{(3,0,1)} \\
\cline{2-3}
\multicolumn{1}{c}{B} &\multicolumn{1}{|c}{(3,0,0)} &
\multicolumn{1}{|c|}{(0,3,0)}\\
\cline{2-3}
\end{array}
$$

This game has a unique equilibrium, $\bar \s$, in which player 1 mixes uniformly between $(T,t)$ and $(T,b)$; player 2 mixes uniformly between $(L,l)$ and $(L,r)$; and player 3 plays $W$.
If we delete $B$, $R$, $E^w$, $E^e$, all of which are unused best replies, we get the game $G_1$:
$$
\begin{array}{ccc}
\multicolumn{1}{c}{} &\multicolumn{1}{c}{l} &\multicolumn{1}{c}{r} \\
\cline{2-3}
\multicolumn{1}{c}{t} &\multicolumn{1}{|c}{(6,6,1)} &
\multicolumn{1}{|c|}{(0,0,1)} \\
\cline{2-3}
\multicolumn{1}{c}{b} &\multicolumn{1}{|c}{(0,0,1)} &
\multicolumn{1}{|c|}{(6,6,1)}\\
\cline{2-3}
\end{array}
$$
where the equilibrium $\bar \s$ is now regular but reverses its index to  $-1$.

Our final example shows that there are problems even when an equilibrium is unique and in pure strategies: deleting some unused best replies results in a game where this equilibrium is not even isolated.  In the game $\bar G$ below
 $$W: \begin{array}{ccc}
\multicolumn{1}{c}{} &\multicolumn{1}{c}{L} &\multicolumn{1}{c}{R} \\
\cline{2-3}
\multicolumn{1}{c}{T} &\multicolumn{1}{|c}{(1,1,1)} &
\multicolumn{1}{|c|}{(1,1,0)} \\
\cline{2-3}
\multicolumn{1}{c}{B} &\multicolumn{1}{|c}{(1,0,1)} &
\multicolumn{1}{|c|}{(0,1,1)}\\
\cline{2-3}
\end{array}\qquad
E:\; \begin{array}{ccc}
\multicolumn{1}{c}{} &\multicolumn{1}{c}{L} &\multicolumn{1}{c}{R} \\
\cline{2-3}
\multicolumn{1}{c}{T} &\multicolumn{1}{|c}{(0,1,1)} &
\multicolumn{1}{|c|}{(1,0,1)} \\
\cline{2-3}
\multicolumn{1}{c}{B} &\multicolumn{1}{|c}{(1,0,0)} &
\multicolumn{1}{|c|}{(0,1,0)}\\
\cline{2-3}
\end{array}
$$
$(T, L, W)$ is the unique equilibrium.\footnote{This is a mild strengthening of the example in Brandt and Fischer \cite{BF2008}, where the unique equilibrium---which also fails to be quasi-strict---was in mixed strategies.} 
Deleting the unused replies $B$ and $E$, but not $R$, gives us the following game $G_2$:

\begin{center}
	\begin{tabular}{ccc}
		& $L$ &  $R$ \\ \cline{2-3}
		$T$ & \multicolumn{1}{|c}{$(1, 1, 1)$} & \multicolumn{1}{|c|}{$(0, 1, 0)$}  \\
		\cline{2-3}
	\end{tabular}
\end{center}
where now there is an interval of equilibria.   

It is noteworthy that, when we view the last game $G_2$ as a two-player game, the equilibrium $(T,L)$ (part of a connected component) is made unique by adding players and strategies. Similarly, in the two-player game $G_1$ above, the $-1$ index regular equilibrium in which player 1 mixes uniformly between $(T,t)$ and $(T,b)$ and player 2 mixes uniformly between $(L,l)$ and $(L,r)$ has been made unique by adding players and strategies, leaving open the question of how general such a property is.\footnote{We thank Joseph Hofbauer for raising this issue.} Observe that without adding players, it is impossible to make a non $+1$ equilibrium unique in a two player game by adding only strategies, because a unique equilibrium of a two player game is necessarily quasi-strict (Norde \cite{N1999}).

\section{Extending Sustainability to Non-generic Games} 
It is fairly easy to construct non-generic games where no equilibrium is sustainable.  For example, consider the (in effect) two-player game $G_2$ we obtained at the end of the last section, where the set of equilibria is a connected component.  There is no way to make any of these equilibria unique by adding strategies that are inferior replies.\footnote{A non isolated equilibrium $\s$ of a game $G$ can never be made unique in a game $\hat G$ obtained from $G$ by adding/deleting inferior replies. Suppose this can be done. Let $(\s^k)_{k \in \mathbb{N}}$ be a sequence of equilibria of $G$ converging to $\s$. Passing to a subsequence if necessary, we can assume that the support $T$ of $\s^k$ is fixed and that all strategies in the support are best replies to $\s$. As such no strategy in $T$ can be deleted. Since $\s$ is the unique equilibrium of $\hat G$, there is (up to a subsequence) a fixed player $i$ and a fixed profitable best reply deviation $\t_i$ to $\s^k$ in $\hat G$. By continuity, $\t_i$ is a best reply to $\s$ in $\hat G$: a contradiction.}  Therefore, any extension of sustainability to non-generic games  needs to allow for solutions to be sets of equilibria, if we are to have existence. It is then natural to require that  solutions be connected: as argued by Mertens \cite{Mertens1989}, this requirement keep our selections ``minimal'' and, in addition, if the game is a generic extensive-form game, each solution would have a unique prediction in terms of the equilibrium outcome.  

A natural definition of sustainability for connected sets of equilibria calls for such a set to be sustainable if it is the entire set of equilibria of a game obtained by adding or deleting strategies that are inferior replies to {\it every} equilibrium in the set.  In such a definition, the component in which the candidate solution lies would survive these alterations and, hence, the connected sets would have to be entire components.\footnote{A strict subset $U$ of an equilibrium component $C$ can never be made the unique equilibrium set in a game $\hat G$ obtained from $G$ by adding/deleting inferior replies. The proof is the same as in the last footnote.}  Moreover, a sustainable component should also have index $+1$ (for the same reason as why sustainability for individual equilibria implies that their index is $+1$).  As there are games\footnote{\cite{HH2002} Fig. 8, or \cite{R2002} p. 325.}  where no component has index $+1$, there seems to be no reasonable way to extend sustainability to non-generic games without running afoul of the existence criterion.

Our way around this problem is to extend sustainability by an axiomatic procedure, rather than by a definition using the game-theoretic property that it attempts to capture.  To do that, we first axiomatize sustainability for regular equilibria and then add axioms that are to hold on the universal domain of finite games to  obtain a characterization of components with positive index.  Actually, the axiomatization of sustainability holds for a slightly larger class of games than  regular games and is derived from the following observation.  In the main theorem of the paper, even though we focus on regular equilibria, what really matters, as a careful reading of the proof shows, is one particular implication of regularity, namely that regular equilibria are isolated. Thus, we have proved a slightly stronger result: 

\begin{theorem}
	An isolated equilibrium of a finite game is sustainable iff its index is $+1$. 
\end{theorem}

Let $\calG$ be the class of all finite games and let $\calG^*$ be the subset consisting of all games where each equilibrium is isolated and has either index $+1$ or index $-1$. $\calG^*$ includes all regular games, but it does not coincide with the set of all games with finitely many equilibria. The following game $G_3$ has two isolated equilibria, neither of which is $+1$ or $-1$.

 $$W: \begin{array}{cccc}
\multicolumn{1}{c}{} &\multicolumn{1}{c}{L} &\multicolumn{1}{c}{M}  &\multicolumn{1}{c}{R} \\
\cline{2-4}
\multicolumn{1}{c}{A} &\multicolumn{1}{|c}{(0,3,3)} &
\multicolumn{1}{|c|}{(0,3,3)} &\multicolumn{1}{|c|}{(3,0,3)}  \\
\cline{2-4}
\multicolumn{1}{c}{B} &\multicolumn{1}{|c}{(3,3,0)} &
\multicolumn{1}{|c|}{(3,3,0)} &\multicolumn{1}{|c|}{(3,3,3)} \\
\cline{2-4}
\multicolumn{1}{c}{C} &\multicolumn{1}{|c}{(6,6,3)} &
\multicolumn{1}{|c|}{(2,2,3)} &\multicolumn{1}{|c|}{(0,7,3)} \\
\cline{2-4}
\multicolumn{1}{c}{D} &\multicolumn{1}{|c}{(2,2,3)} &
\multicolumn{1}{|c|}{(6,6,3)} &\multicolumn{1}{|c|}{(0,0,3)} \\
\cline{2-4}
\multicolumn{1}{c}{E} &\multicolumn{1}{|c}{(0,0,3)} &
\multicolumn{1}{|c|}{(7,0,3)} &\multicolumn{1}{|c|}{(1,1,3)} \\
\cline{2-4}
\end{array}\qquad
O:\; \begin{array}{cccc}
\multicolumn{1}{c}{} &\multicolumn{1}{c}{L} &\multicolumn{1}{c}{M}  &\multicolumn{1}{c}{R} \\
\cline{2-4}
\multicolumn{1}{c}{A} &\multicolumn{1}{|c}{(0,3,3)} &
\multicolumn{1}{|c|}{(0,3,3)} &\multicolumn{1}{|c|}{(3,0,0)}  \\
\cline{2-4}
\multicolumn{1}{c}{B} &\multicolumn{1}{|c}{(3,0,3)} &
\multicolumn{1}{|c|}{(3,0,3)} &\multicolumn{1}{|c|}{(0,3,3)} \\
\cline{2-4}
\multicolumn{1}{c}{C} &\multicolumn{1}{|c}{(6,6,0)} &
\multicolumn{1}{|c|}{(2,2,0)} &\multicolumn{1}{|c|}{(0,7,0)} \\
\cline{2-4}
\multicolumn{1}{c}{D} &\multicolumn{1}{|c}{(2,2,0)} &
\multicolumn{1}{|c|}{(6,6,0)} &\multicolumn{1}{|c|}{(0,0,0)} \\
\cline{2-4}
\multicolumn{1}{c}{E} &\multicolumn{1}{|c}{(0,0,0)} &
\multicolumn{1}{|c|}{(7,0,0)} &\multicolumn{1}{|c|}{(1,1,0)} \\
\cline{2-4}
\end{array}
$$
\medskip

Game $G_3$ has two isolated equilibria, one mixed $\sigma=(\frac{1}{2}C+\frac{1}{2}D;\frac{1}{2}L+\frac{1}{2}M,W)$, and one pure $\tau=(B,R,W)$. It is easy to check that $\sigma$ is quasi-strict, and so one can compute that its index is $-1$ by a restriction to its support. Thus, the index of  $\tau$ is $+2$ and so no equilibrium in $G_3$ can be made unique in an equivalent pair.\footnote{Our game is inspired by the two player game in \cite{HH2002} fig 8 in which the $+2$ index was a component. By adding a player and strategies we can kill all equilibria of that component except $\tau$.}

A solution concept $\Phi$ on the domain $\calG^*$ assigns to each game $G \in \calG^*$ a collection of equilibria (called {\it solutions of $G$}).  Let $\Phi^*$ be the solution concept that assigns to each game $G$ in $\calG^*$ the equilibria with index $+1$.  $\Phi^*$ is then the unique solution concept that satisfies the following axioms for a solution concept $\Phi$ on the domain $\calG^*$.\footnote{These axioms have been explicitly or implicitly stated in Myerson \cite{M1996} and Hofbauer \cite{H2000}.  In fact, Hofbauer's conjecture is based on a combination of A1, A2 and A3.}

\begin{itemize}
\item \textbf{$A1$ Existence:}  Every game in $\calG^*$ has a solution.
\item \textbf{$A2$ IIA:} If $(G^1, \s^1) \sim (G^2, \s^2)$, where for $i = 1, 2$, $G^i \in \calG^*$, then $\s^1 \in \Phi(G^1)$ iff $\s^2 \in \Phi(G^2)$.  
\item \textbf{$A3$ Minimality}: If $\bar \Phi$ is another solution concept satisfying the first two axioms, then every solution assigned by $\Phi$ is a solution of $\bar \Phi$.
\end{itemize}

Consider now the following axioms for an extension $\Phi$ of $\Phi^*$ to $\calG$

\begin{itemize}
\item\textbf{$A1^+$ Existence}: Every game $G \in \calG$ has a solution.
\item\textbf{$A2^+$ IIA}: A solution of a game $G$ is also a solution of any game $\bar G$ obtained from the game $G$ by the addition and/or deletion of strategies that are inferior replies to every equilibrium in the solution.
\item \textbf{$A4$ Connectedness}: Every solution is a connected subset of Nash equilibria.
\item \textbf{$A5$ Invariance}: Equivalent games have equivalent solutions.\footnote{Two games are equivalent if one can be obtained from the other by addition/deletion of duplicate of mixed strategies.}
\item \textbf{$A6$ Robustness}: Every game that is nearby in the space of payoffs has a nearby solution.
\end{itemize}

Axioms $A1^+$ and $A2^+$ are the natural statements of the corresponding axioms for the domain $\calG$. Observe that we dispense with a counterpart of Axiom 3 since we are not directly extending sustainability but only its prescription for generic games. Axioms A4 and A5 are standard in the strategic stability literature (Kohlberg and Mertens \cite{KM1986}).  We have already discussed the reasonableness of Axiom $A4$.  Axiom A5 requires that the solution depend only on the reduced normal form and thus be invariant to irrelevant changes in the extensive form description of the game or to the addition/deletion of duplicate strategies (mixtures of pure strategies). Axiom A6 is really a restatement of hyperstability  \cite{KM1986}. 
The next proposition shows that there is a solution concept extending $\Phi^*$ to $\calG$ and satisfying all our axioms, namely the positive index Nash components. Afterwards it is proved that it is the unique concept satisfying $A1^+$, $A4$ and a combination of $A5$ and $A6$\textit{}. 

\begin{proposition}\label{Prop positive index}
The solution concept $\Phi^+$ that associates to each game its set of positive index Nash components extends $\Phi^*$ to $\calG \backslash \calG^*$ and satisfies $A1^+$, $A2^+$, $A4$, $A5$, and $A6$.
\end{proposition}

\begin{proof} Obviously, the restriction of $\Phi^+$ to $\calG^*$ is $\Phi^*$ and is thus an extension to $\calG$.   $A1^+$ holds because the sum of indices across the finitely many Nash components is +1 (Ritzberger \cite{R1994}).  Axiom (A4) holds as positive index Nash components are closed and connected subsets of Nash equilibria.  Axiom (A5) is a consequence of the invariance of the index to the addition/deletion of duplicate strategies (Govindan and Wilson \cite{GW2005}, Theorem 5), the same argument also proves $A2^+$ as well.  Axiom A6 is a consequence of the fact that if the index is nonzero, then the component of equilibria is essential in the sense of O'Neill \cite{O1953}. 
\end{proof}

The solution concept that to games in $\calG^*$ assigns equilibria with index +1 and to other games all components with nonzero index is also an extension. It is not clear if our axioms rule out such solution concepts.  However, we now show that if we strenghthen  Axioms A5 and A6---in a sense, conflating them---then we obtain a characterization of $\Phi^+$. 

We will now impose this property for an extension $\Phi$ of $\Phi^*$.
\begin{itemize}
	\item \textbf{$A7$  Uniform Robustness:} For every game $G$, every solution $C$ and  every neighbourhood $U$ of $C$, there exists $\delta >0$ such that every $\delta$-perturbation of the payoffs in every strategically equivalent game $G'$ yields to a game $G'_{\delta}$ that has a solution that is equivalent to a subset in $U$.
\end{itemize}

If the  term ``solution'' is replaced by ``Nash equilibrium'' we obtain precisely the uniform hyperstability concept defined in Govindan and Wilson \cite{GW2005} where they proved that non-zero index Nash components are the only connected uniformly hyperstable sets. Adapting their tools allows to show the following axiomatic characterization of positive index Nash components.\footnote{A closed subset of Nash equilibria $C$ of a finite game $G$ is hyperstable if for any strategically equivalent game $G'$ and every neighbourhood $U'$ of the equivalent set $C'$ of equilibria, there exists a neighbourhood $V'$ of $G'$ such that every game in $V'$ has a Nash equilibrium in $U'$.  We conjecture that non-zero index Nash components are the only connected hyperstable sets in generic extensive form games, giving us the analogous conjecture in our context: A7 is implied by A5 and A6 for generic extensive-form games.} This theorem suggests that in non-regular games, sustainable equilibria are the components of equilibria with positive index. 

\begin{theorem}\label{mainthm}
$\Phi^+$ satisfies A7.  Moreover, any extension of $\Phi^*$ to $\calG$ that satisfies A4 and A7 must select from among the solutions of $\Phi^+$.
\end{theorem}

\begin{proof} If a solution $C$ for a game $G \notin \calG^*$ is not an entire component, then the proof in Govindan and Wilson \cite{GW2005} applies to show that $C$ cannot be uniformly robust. Therefore, a solution must be an entire connected component of equilibria.  If $C$ is a component with negative index, then Lemma \ref{lemma1} in the Appendix implies that  there is a neighbourhood $U$ of $C$ such that for every $\delta>0$, there exists an equivalent game $\bar G$ to $G$ and a $\d$-perturbation $\bar G^{\d}$ of $\bar G$ with finitely many equilibria of index $+1$ and $-1$ such that all the equilibria of $\bar G^{\d}$ in $\bar U$ (the equivalent neighbourhood to $U$) have index $-1$ and so are not ``solutions''. Thus $C$ must have a positive index.  Conversely, that positive index components are uniformly robust follows from Govindan and Wilson \cite{GW2005}. \end{proof}

\section{Conclusion}

Most existing \textit{refinement} concepts that are guaranteed to exist, such as perfection \cite{S1975}, properness \cite{M1978},  stable equilibria \cite{KM1986, Mertens1989, H1990} only refine equilibria of nongeneric games, i.e., every equilibrium of a regular game survives these refinement criteria.\footnote{An exception is persistent equilibrium \cite{KS1984} which eliminates the mixed equilibrium in the battle of the sexes; but, as noted by Myerson, it violates invariance \cite{B1992}.}  On the other hand, the procedure proposed by Harsany and Selten \cite{HS1988} \textit{selects} a unique equilibrium in every game.  Sustainability lies somewhere in between refinement and selection: it slices by half the set of predictions in regular games, as it disregards all -1 index equilibria.

In nongeneric games,  one argument against components with negative index is that they are dynamically unstable for all Nash dynamics \cite{DG2000, M2018}.  Our approach provides an axiomatic argument against these components.  Moreover, the selection among components with a positive index can be made compatible with all of the above-mentioned refinements. Indeed, it is known that components with nonzero index always contain stable sets in the sense of Mertens (and hence also those that are stable in the sense of Kohlberg-Mertens or Hillas) and thus contain proper equilibria  and sequential equilibria of every extensive form game with that normal form \cite{KM1986, vD1987}. In addition, there are positive index components that are persistent.\footnote{This is because the best reply correspondence maps a neighbourhood of a persistent retract to it, so the sum of the indices in a retract is $+1$ and so there is at least one component in the persistant retract with a positive index.}

We alluded to a number of open problems in the paper.  One of the most important of these is the question of whether we can retain axioms $A5$ and $A6$, and eliminate $A7$---at least on the domain of generic games in extensive-form---to select components with a positive index. A second question that is intriguing  is the effect of adding more players to the game.  For example, we can view an $N$-person game as, for example, an $(N+1)$-person game by treating player $N+1$ as a dummy player and then considering equivalences of game-equilibrium pairs involving these $N+1$ players. How would this more expansive notion of equivalence compare with what we know so far? Finally, one wonders whether the selection from among the components with a positive index by invoking some additional criterion like persistent equilibria \cite{KS1984} or settled equilibria \cite{MW2015} could lead to an equilibrium notion that  addresses some of the shortcomings of using only the positivity of index.  Grounds for optimism here come chiefly from examples like game $G_1$ in the Introduction, where the completely mixed equilibrium has index $+1$ and is unstable for all natural dynamics, but it is neither persistent nor settled.

\appendix
\section{Delaunay Triangulations}

Here we construct a triangulation $\T_n$ of $\S_n \times \Theta_n$ for each $n$ with the properties stated in subsection \ref{lastembedding}. We start with some definitions.

A {\it simplex} $T$ in $\Re^d$ is the convex hull of affinely independent points $x_0, x_1, \ldots, x_k$ $(k \le d$); a {\it face} of $T$ is the convex hull of a subset of the points $x_i$. A {\it triangulation} $\T$ of a polytope $C \subset \Re^d$ is a finite collection of simplices $T$ in $\Re^d$ such that: (1) if $T \in \T$, so is every face of $T$; (2) the intersection of two simplices in $\T$ is a face of both (possibly empty); (3) the union of the simplices in $\T$ equals $C$.

Throughout this Appendix, we use the $\ell_2$-norm, unless we specify differently. Take a finite collection of points $\{\, x_0, x_1, \ldots, x_k \, \}$ in $\Re^d$ (where $k$ is now an arbitrary positive integer) such that its convex hull $C$ is $d$-dimensional.  Suppose that the $x_i$'s are  in \textit{general position for spheres} in $\Re^d$---i.e., no subcollection of $d+2$ of these points lie on any $(d-1)$-sphere (centered at any point and of any radius) in $\Re^d$. We can construct a triangulation of the convex hull $C$, called the {\it Delaunay triangulation}, as follows (cf. Loera et al \cite{LRS2010} for details). Let $D$ be the convex hull of the set of points $(x_i, {\Vert x_i \Vert}^2) \in \Re^{d+1}$, $i= 0,1, \ldots k$.  Let $D_0$ be the lower envelope of $D$. The natural projection $(x,y) \mapsto x$ from $D_0$ to $\Re^d$ is $C$ and $D_0$ is the graph of a piece-wise linear and convex function $\rho: C \to \Re$ with the property that the subsets on which $\rho$ is linear are simplices, whose projections then yield the simplices of a triangulation of $C$. 

There is a dual representation of the Delaunay triangulation, known as the {\it Voronoi Diagram}, which works as follows. For each $i= 0,1, \ldots k$, let $P_i$ be the polyhedron in $\Re^d$ consisting of points $y$ in $\Re^d$ such that ${\Vert y - x_i \Vert} \le {\Vert y - x_j \Vert}$ for all $j \neq i$. We then have a polyhedral complex (which is like a simplicial complex but with polyhedra rather than simplices) where the maximal polyhedra are the $P_i$. There is an edge between two vertices $x_i$ and $x_j$ in the Delaunay triangulation iff the polyhedra $P_i$ and $P_j$ have a nonempty intersection. Also, the intersection of $d+1$ of these polyhedra when nonempty is a single point (because of genericity), which is then the center of a ball that contains $d+1$ points of the collection on its boundary and no other point in the ball itself---these $d+1$ points span a $d$-dimensional simplex in the Delaunay triangulation.

For our purposes, we need a triangulation with the diameter of certain simplices to be smaller than $\zeta$, as specified in Subsection \ref{lastembedding}. To obtain that, we need an auxiliary construction. Let $C$ be a full-dimensional polytope in $\Re^d$.  Let $B_0$ be a proper face of $C$ and let $H$ be a hyperplane that strictly separates $B_0$ from the vertices of $C$ that are not in $B_0$.  Let $B$ be the intersection of $C$ with the halfspace generated by $H$ that contains $B_0$ in its interior, i.e., $B$ is of the form $C \cap H_{-}$ where $H_{-} = \{\, x \in \Re^d \mid a \cdot x \le b \,\}$ and $a \cdot x < b$ for all $x \in B_0$. Let $X$ be the set of vertices of $C$ and let $Y$ be the set of vertices of $B$ that are not in $X$.  Suppose the vertices in $X \cup Y$ satisfy the following property $(P)$: for each $0 < k \le d$, and each collection  $v_0, \ldots, v_{k+1}$ of vertices in $X \cup Y$ that are not affinely independent,  the intersection of the Voronoi polyhedra of these $k+2$ vertices is empty. This assumption implies, a.o., by taking $k = d$, that the vertices in $X \cup Y$ are in general position for a Delaunay triangulation of $C$. Beyond that, the stronger assumption allows us to obtain refined Delaunay triangulations of $C$. Specifically, if we add a collection $Z$ of points in $C$ such that if each point $z \in Z$ is in generic position in the face of $C$ or $H\cap C$ that contains it in its interior---i.e., outside a nowhere dense subset of this face---then the collection $X \cup Y \cup Z$ is in general position for spheres and we can construct a Delaunay triangulation with this set of vertices. Indeed, suppose the points in $Z$ are chosen generically and suppose $v_0, \ldots, v_{d+1}$ is a collection of vertices in $X \cup Y \cup Z$.  There exists $0 < k \le d$ such that, after permuting the labels of the collection if necessary, $v_0 = \sum_{i = 1}^{k+1} \l_i v_i$, $\l_i \neq 0 $ for all $i > 0$ and $\sum_i \l_i = 1$. By Property (P), the Voronoi polyhedra of the vertices  do not intersect if all the vertices belong to $X \cup Y$.  If there is some vertex in the collection that belongs to $Z$, then we can assume that $v_0$ belongs to $Z$. If $v_1, \ldots, v_{k+1}$ do not span a face of $C$ or $H \cap C$, then clearly we can perturb $v_0$ so that it does not lie in the affine space of the other $v_i$'s, which contradicts the assumption that $v_0$ was chosen in a generic position.  On the other hand, if $v_1, \ldots, v_{k+1}$ do span a face of $C$ or $H \cap C$, then for generic choice of $v_0$ and also of the $v_i$'s in the list $v_1, \ldots, v_{k+1}$ that are not in $X \cup Y$, their Voronoi polyhedra do not intersect.  Thus, Property (P) allows to us refine the initial Delaunay triangulation of $C$ (involving vertices $X \cup Y$).

Let $\d> 0$ be such that ${\Vert x - y \Vert} \ge \d/2$, for all $x \in B$ and vertices $y \in C \backslash B$. Let $X_\d$ be a finite collection of points in $C$ such that: (1) $X \cup Y \subset X_\d$ and  $X_\d \cap (C \backslash B) \subset X$; (2) for $x \in B$, there is a point $x_\d  \in B \cap X_\d$ such that ${\Vert x - x_\d \Vert} < \d/2$ and $x_\d$ belongs to the face of $B$ that contains $x$ in its interior; (3) every point in $\text{int}_C(B) \cap X_\d$ is at least $\d/2$ from $\partial_C B$;  (4) the points in $X_\d$ are in general position for spheres.  Call $\T_\d$ the associated Delaunay triangulation of $C$.  

The triangulation  $\T_\d$ above achieves two properties: (i) every simplex with vertices in $B$ has diameter at most $\d$; (ii) every simplex of $\T_\d$ that has a vertex outside $B$ does not intersect $\int_{C}(B)$. To prove these properties, define $r: \Re^d \to B$ by letting $r(x)$ be the point in $B$ that is closest to $x$. If $r(x) \neq x$, $r(x)$ belongs to a proper face of $B$, and then we can write  $r(x)$ as $x - p$ where $p$ is a normal for a supporting hyperplane at $r(x)$  with $p \cdot r(x) \ge p \cdot y$ for all $y \in B$. If in addition $r(x) \in \int_C(B)$, then $r(x)$ is at the boundary of $C$ and so $p\cdot r(x) \ge p \cdot y$ for all $y \in C$ as well. Suppose $r(x) \neq x$ and let $r(x) = x - p$.  Let  $y$ be a point such that $p\cdot y \le p \cdot r(x)$. Let $z$ be the nearest-point projection of $y$ onto the line from $x$ through $r(x)$. Then 
\[
{\Vert x - y\Vert }^2 = {({\Vert x - r(x) \Vert} + {\Vert r(x) - z \Vert})}^2 + {\Vert z  - y\Vert}^2 \ge {\Vert r(x) - x \Vert}^2 + {\Vert r(x) - y \Vert}^2,
\] 
with the inequality being an equality iff $z = r(x)$, i.e., $p\cdot y = p\cdot r(x)$.

We are now ready to prove that $\T_\d$ has the requisite properties. Let $x_\d$ be a point in $X_\d \cap B$ and let $x$ be a point in $\Re^d$ that belongs to the Voronoi polyhedron $P(x_\d)$ of $x_\d$. We claim that ${\Vert r(x) - x_\d \Vert} < \d/2$.  If $r(x) = x$, this follows directly from Property (2) of $X_\d$. Suppose that $r(x) \neq x$. Then $r(x)$ belongs to the interior of a proper face $B'$ of $B$ and as we saw in the last paragraph, $r(x)$ can be written as $x - p$.  By definition of $r(x)$, $p \cdot x_\d \le p \cdot r(x)$ and thus:  ${\Vert x - x_\d\Vert }^2  \ge {\Vert r(x) - x \Vert}^2 + {\Vert r(x) - x_\d\Vert}^2$.   
By Property (2), there exists $y_\d$  in $B' \cap X_\d$ such that $\Vert r(x) - y_\d\Vert <\d/2$.  Obviously $p\cdot y_\d = p \cdot r(x)$ and since $x \in P(x_\d)$, it follows that ${\Vert x - x_{\d} \Vert}^2 \le {\Vert x - y_\d \Vert}^2 < {\Vert r(x) - x \Vert}^2 + \d^2/4$; therefore, ${\Vert r(x) - x_\d \Vert} < \d/2$, as claimed. Observe that we proved that ${\Vert x - x_\d \Vert}^2 < {\Vert r(x) - x \Vert}^2 + \d^2/4$, a fact we will use below.

From the above paragraph, for each $x_\d \in X_\d \cap B$ and each $x \in P(x_\d)$, the distance between $r(x)$ and $x_\d$ is less than $\d/2$; We claim finally that the diameter of each simplex with vertices in $B$ is less than $\d$. Indeed, letting $x_\d$ and $y_\d$ be two vertices of a simplex in $B$, then their Voronoi cells intersect, so we can take $x$ in the intersection. Since $r(x)$ is of distance $\d/2$ from $x_\d$ and $\d/2$ from $y_\d$, $x_\d$ and $y_\d$ are distant less than $\d$. This concludes the proof that $\T_{\d}$ satisfies  (i). 

We now prove that $\T_\d$ satisfies (ii): for this, it is sufficient to show  that the intersection of $P(x_\d)$ and $P(y_\d)$ is empty for all $x_\d \in \int_C (B) \cap X_\d$ and $y_\d$ in $X_\d \backslash B$. Take such a pair $x_\d, y_\d$.  Fix  $x \in P(x_\d)$. If $r(x) = x$, then ${\Vert x - x_\d \Vert} < \d/2$, while by the definition of $\d$, ${\Vert x - y_\d \Vert} \ge \d/2$  and thus $x \notin P(y_\d)$.
Suppose $r(x) \neq x$.  Since $x_\d \in \int_C(B)$, by Property (3) of the set $X_\d$,  $r(x)$ cannot belong to $\partial_C B$, since in this case the distance between $x_\d$ and $r(x)$ is greater than $\d/2$. Therefore, $r(x)$ belongs to a face of $C$.  Writing $r(x)$ as $x - p$, we then have  that $p$ is a normal to a hyperplane containing one of the faces of $C$ and thus $p\cdot y_\d \le p\cdot r(x)$. Hence, ${\Vert x - y_\d\Vert }^2  \ge {\Vert r(x) - x \Vert}^2 + {\Vert r(x) - y_\d\Vert}^2 \ge {\Vert r(x) - x \Vert}^2  + \d^2/4$ by the definition of $\d$, while as we saw in the previous paragraph, ${\Vert x - x_\d \Vert}^2 < {\Vert r(x) - x \Vert}^2 + \d^2/4$; thus again $x \notin P(y_\d)$ and we are done.

For our problem of triangulating $\S_n \times \Theta_n$, we recall the properties that the triangulation should satisfy: (1) The only vertices in $\S_n \times \{\, \theta_n^0 \, \}$ of $\T_n$ are pure strategies $(s_n, \theta_n^0)$, $s_n \in S_n$; (2) letting $\Theta_n^1$ be the face of $\Theta_n$ where $\theta_n^0$ has zero probability, if $T_n\in \T_n$ is a simplex either with a face in $\S_n \times \Theta_n^1$, or shares a face with such a simplex,  then the diameter of $T_n$ is less than $\zeta$; (3) there exists a convex function $\rho_n: \S_n \times \Theta_n \to \Re_+$ such that: (a) $\rho_n(\l x + (1-\l) y) = \l \rho_n(x) + (1-\l) \rho_n(y)$ iff $x$ and $y$ belong to a simplex $T_n$ of $\T_n$; (b) $\rho_n^{-1}(0) = \S_n \times \{\, \theta_n^0 \, \}$. 

Let $\hat S_n^1 \equiv S_n \times \{\, \theta_n^0 \, \}$. Let $T_n$ be the set of vertices of $\Theta_n$ and $\hat S_n^2  \equiv \{\, s_n^0 \, \} \times (T_n \backslash \{\, \theta_n^0\})$. Let $\hat S_n \equiv \hat S_n^1 \cup \hat S_n^2 \setminus \{(s^0_n, \theta^1_n) \}$ where $\theta^1_n$ is a vertex of $\Theta_n$ that is different from $\theta^0_n$. Note that $d \equiv \text{dim}(\S_n \times \Theta_n) = |S_n| + |T_n| -2 = |\hat S_n|$. For each $\hat s^1_n \in \hat S_n^1$, let $x(\hat s_n^1)$ be the unit vector in $\Re^{\hat S_n}$ for the coordinate $\hat s_n^1$; for each $\hat s^2_n \in \hat S_n^2$, let $x(\hat s_n^2)$ be a point in $\Re^{\hat S_n}$ to be determined later. Let $X \equiv \{x(\hat s^2_n)\}_{\hat s^2_n \in \hat S^2_n}$.

Define an affine function $F_n^X :  \S_n \times \Theta_n \to \Re^{\hat S_n}$ as follows:  for each $\hat s_n \in \hat S_n$, $F_n^X (\hat s_n) = x(\hat s_n)$; for a vertex $(s_n, \theta_n)$ of $\S_n \times \Theta_n$ that is not in $\hat S^1_n \cup \hat S^2_n$, define $F_n^X (s_n, \theta_n) = F_n^X(s_n, \theta_n^0) + F_n^X(s_n^0, \theta_n) - F_n^X(s_n^0, \theta_n^0)$. The map $F^X_n$ extends to the whole of $\S_n \times \Theta_n$ by linear interpolation.  If the collection $X \cup \{x(\hat s^1_n)\}_{\hat s^1_n \in \hat S^1_n}$ is affinely independent (which holds for an open and dense set of choices for $X$), then $F_n^X$ is an affine homeomorphism with its image $C(X) \equiv F^X_n(\S_n \times \Theta_n)$ and the dimension of $C(X)$ is $|\hat S_n|$.  

Let $H$ be a hyperplane in $\Re^{S_n \times T_n}$ that strictly separates $\S_n \times \Theta_n^1$ from $\S_n \times \{\, \theta_n^0\, \}$ and such that the distance between $H \cap (\S_n \times \Theta_n)$ and $\S_n \times \Theta_n^1$ is less than $\zeta$.  Every vertex of $H \cap (\S_n \times \Theta_n)$ is of the form $(1-\e_{s_n, \theta_n})(s_n, \theta_n) + \e_{s_n, \theta_n}(s_n, \theta_n^0)$ for some $s_n \in S_n$, $\theta_n^0 \neq \theta_n \in \Theta_n$ and $0 < \e_{s_n, \theta_n} < \zeta$.  The hyperplane could be defined equivalently by choosing $|S_n| + |T_ n|$ of the coordinates $\e_{s_n, \theta_n}$.  Let $B_0(X) = F_n^X(\S_n \times \Theta_n^1)$ and let $B(X) = F_n^X(H_- \cap (\S_n \times \Theta_n))$, where $H_-$ is the halfspace that contains $\S_n \times \Theta_n^1$.  

Let $\bar X$ be the set of vertices of $C(X)$ and let $\bar Y$ be the set of vertices of $C(X) \cap F_n^X(H \cap (\S_n \times \Theta_n))$.  We claim now that if the choice of the vectors in $X$ is generic and if the hyperplane $H$ is in generic position, i.e., if the collection of $\e_{s_n, \theta_n}$ is chosen outside a nowhere dense set, then $\bar  X \cup \bar Y$ satisfies property $(P)$; consequently there exists a Delaunay triangulation of $C(X)$ that can be refined according to $(P)$. To prove this claim, let $v_0, \ldots, v_{k+1}$ be a collection of vertices  each belonging to either  $\S_n \times \Theta_n$ or its intersection with $H$ and such that $v_0 = \sum_{i = 1}^{k+1} v_i$, $\l_i \neq 0$ for all $i > 0$ and $\sum_i \l_i = 1$.  We have to show that the  Voronoi polyhedra of the $v_i$'s do not intersect. We divide the proof in three cases: assume to begin with that all the vertices belong to $\S_n \times \Theta_n$.  Then, $k+2 = 2J$ for some integer $J > 1$ and, after a relabeling of the vertices, $\sum_{i = 1}^{J} v_i = \sum_{i = J+1}^{2J} v_i$; of these there are as many vertices in $\hat S_n^1$ among the first $J$ as there are in the second; and there is at least one vertex on each side of the equality that does not belong to $\hat S_n^1$. If the Voronoi polyedra of the $F_n^X(v_i)$'s contain a point $y$ in common, then letting $c$ be the distance between $y$ and each of the $F_n^X(v_i)$'s, elementary algebra shows that $2F_n^X(v_i)\cdot y = \Vert F_n^X(v_i) \Vert + \Vert y \Vert - c^2$ for each $i$. Since $F^X_n$ is affine, it follows that $\sum_{i = 1}^J{\Vert F^X_n(v_i) \Vert}^2  = \sum_{i = J+1}^{2J}{\Vert F^X_n(v_i) \Vert}^2$, an equality that holds only for a nongeneric choice of $X$. Therefore,  the Voronoi polyhedra of the $F_n^X(v_i)$'s do not intersect. 

Now suppose that all the $v_i$'s belong to $H \cap C$. Then the nearest-point projection of the $v_i$'s to $\S_n \times  \Theta_n^1$ span a face. The Voronoi polyhedra of the projections do not have a point in common, as we saw in the previous paragraph. Therefore, if the $\e_{s_n, \theta_n}$'s are small, the Voronoi polyhedra of the $v_i$'s do not intersect either.

Finally, there remains to consider the case where one of them, which we can assume to be $v_0$, belongs to $H$.  There must be at least two $v_i$'s in $\S_n \times \Theta_n$ as the hyperplane $H$ does not contain any vertex of $\S_n \times \Theta_n$. In particular, there are at most $k \le d$ vertices in the collection of $v_i$'s that belong to $H$. Therefore,  we can perturb the $\e_{s_n, \theta_n}$ corresponding to $v_0$ but not the others, to ensure that the Voronoi polyhedra do not intersect. Thus we have verified our claim that for generic $X$ and $H$, the vertices in $\bar X \cup \bar Y$ allows to construct a Delaunay triangulation of $C(X)$ satisfying $(P)$.

Take now a generic set $X$ with the property that the norm of each $x(\hat s_n^2)$ is strictly greater than one for $\hat s_n^2 \in \hat S_n^2$.  Since $F^X_n$ is an affine homeomorphism,  ${\Vert x - y \Vert}_\infty \le M  {\Vert F^X_n(x) - F^X_n(y) \Vert}$ for some $M > 0$ and for all $x, y \in \S_n \times \Theta_n$.  Let $\d>0$ be smaller than $M\zeta$. Using the construction of a subdivision described above, we now have a triangulation  $\T_\d$ of $C(X)$ where each point in  $\int_C (B(X))$ belongs to a simplex with diameter less than $\d/M$, giving us properties (1) and (2) of Subsection \ref{lastembedding}.  As for property (3),  if $\rho_n$ is the convex function associated to the Delaunay triangulation $\T_\d$, the composition $\rho_n \circ F^X_n$ is convex and linear precisely on each cell of the triangulation of $\S_n \times \Theta_n$ induced by the inverse mapping $(F^X_n)^{-1}: C(X) \to \S_n \times \Theta_n$. Our convex function takes value one on $\S_n \times \{\, \theta_n^0 \, \}$ and is strictly above one elsewhere. Subtracting now $1$ from $\rho_n \circ F^X_n$ we have a convex function satisfying property (3).

\section{A Lemma for Section 5}

This Appendix states and proves a key lemma that was invoked in Section 5.  The lemma draws on the following concept of equivalence between normal-form games, which  is generated by the addition and deletion of duplicate strategies.  Say that two normal-form games are \textit{equivalent} if they have the same reduced normal form.  Given a game $G$ with strategy space $\S$, if an equivalent game $\bar G$ with strategy space $\bar \S$ is obtained by adding duplicate strategies, then there is natural (affine) map from $\bar \S$ to $\S$ that sends each profile $\bar \s$ in $\bar G$ to an equivalent strategy $\s$ in $G$; in this case we say that $\bar \s$ \textit{projects to} $\s$.  

Let $C_1, \ldots, C_k$ be the components of Nash equilibria of a finite game $G$.  For each $i$, let $c_i$ be the index of $C_i$. Choose $\e > 0$ such that  the closed $\e$-neighborhoods $U_i$ of the $C_i$'s  are pairwise disjoint.  (All the norms in this Appendix will be $\ell_{\infty}$-norms.) The following lemma is the main result of this Appendix.

\begin{lemma}\label{lemma1}
For each  $\d>0$,  there exist  a game $\bar G$ obtained from $G$ by adding duplicate strategies and a $\d$-perturbation $\bar G^\d$ of $\bar G$ such that: 
\begin{enumerate} 
\item Every equilibrium of $\bar G^\d$ is isolated and projects to a profile in $U_i$ for some $i$;
\item For each $i = 1,...,k$, there are exactly $|c_i|$ equilibria of $\bar G^{\d}$ projecting to $U_i$: if $c_i >0$, they all have index $+1$; if $c_i$ is negative, they all have index $-1$.
\end{enumerate}
\end{lemma}

The proof of the lemma calls upon the concept of multisimplicial complexes and multisimplicial approximations, introduced in \cite{GW2005}, which we review briefly now.  We assume that the reader is familiar with simplicial complexes, which are used as building blocks for multisimplicial complexes. A \textit{multisimplex} is a set of the form $K_1 \times ... \times K_m$, where for each $i$, $K_i$ is a simplex (in some Euclidean space). A \textit{multisimplicial complex} $\K$ is a product $\K_1 \times ... \times \K_m$, where for each $i$, $\K_i$ is a simplicial complex. The \textit{vertex set} $V$ of a multisimplicial complex $\K$ is the set of all $(v_1,...,v_m)$ for which $v_i$ is a vertex of $\K_i$. For a simplicial complex $\L$, denote by $|\L|$ the space of the simplicial complex. The space of the multisimplicial complex $\K$ is then the product space $\prod_{i}|\K_i|$ and is denoted $|\K|$. For each vertex $v$ of $\K$, the star of $v$, denoted $\text{St}(v)$, is the set of all $\s \in |\K|$ such that for each $i, \s_i \in \text{St}(v_i)$. 
A subdivision of a multisimplicial complex $\K$ is a multisimplicial complex $\K^* = \prod_{i} \K^*_i$, where for each $i$, $\K^*_i$ is a subdivision of $\K_i$. 

\begin{definition}\label{multisimplicial}
Let $\K$ be a multisimplicial complex and let $\L$ be a simplicial complex. A map $f: |\K| \to |\L|$ is called \textit{multisimplicial} if for each multisimplex $K$ of $\K$ there exists a simplex $L \in \L$ such that 
\begin{enumerate}
\item $f$ maps each vertex of $K$ to a vertex of $L$;
\item $f$ is multilinear on $|K|$; i.e., for each $\s \in K$, $f(\s) = \sum_{v \in V}f(v) \cdot \prod_{i}\s(v_i)$. 
\end{enumerate} 
\end{definition}

\begin{remark} We call the restriction of a multisimplicial map to its vertex set as a {\it vertex map}. From the definition above, one sees that a multisimplicial map is completely determined by its vertex map. 
\end{remark}

A map  from a multisimplicial complex to another is called multisimplicial if each coordinate map is multisimplicial in the sense of Definition \ref{multisimplicial}.

\begin{definition}Given a multisimplicial complex $\T$ with $|\T| = \S$, a multisimplex $K$ of $\T$ is called \textit{maximal} if the dimension of $K$ equals the dimension of $\S$. Let $\T^*$ be a multisimplicial subdivision of the multisimplicial complex $\T$. Let $h: |\T^*| \to |\T|$ be a multisimplicial function where $|\T^*| = |\T| = \S$. A multisimplex $K^*$ of $\T^*$ is called \textit{fixed} (by $h$) if the lowest dimensional multisimplex $D$ of $\T$ that contains $h(K^*)$ also contains $K^*$.  \end{definition}

\begin{remark}
 If a multisimplex $K^*$ of the map $h$ in the above definition contains a fixed point in its interior, then necessarily it is fixed by $h$, which is the reason for the terminology.  Observe, however, that the converse is not necessarily true.	
\end{remark}

\begin{definition}\label{multi approx}
Let $\K$ be a multisimplicial complex and $\L$ a simplicial complex. Let $g: |\K| \to |\L|$ be a continuous map. A multisimplicial map $f: |\K| \to |\L|$ is a \textit{multisimplicial approximation} of $g$ if for each $\s \in |\K|$, $f(\s)$ is in the unique simplex of $\L$ that contains $g(\s)$ in its interior. 
\end{definition}

The proof of the next claim can be found in the Appendix B of Govindan and Wilson \cite{GW2005}. 

\begin{claim}\label{govwilsonclaim}
Suppose that $g: |\K| \to |\L|$ is a continuous map. There exists $\eta >0$ such that for each subdivision $\K^*$ of $\K$ with the property that the diameter of each multisimplex is at most $\eta$, there exists a multisimplicial approximation $f: |\K^*| \to |\L|$ of $g$.
\end{claim}

\begin{remark}
Let $\T^*$ be a multisimplicial subdivision of a multisimplicial complex $\T$ with $|\T| = \S$. Let $g: |\T^*| \to |\T|$ be a continuous map. We call a multisimplicial map $f: |\T^*| \to |\T|$ a multisimplicial approximation of $g$ if for each $n \in \N$, $f_n: |\T^*| \to |\T_n|$ is a multisimplicial approximation of $g_n$ in the sense of Definition \ref{multi approx}.
\end{remark}

\subsection*{Proof of Lemma \ref{lemma1}} The proof of Lemma \ref{lemma1} is inspired by the idea of the proof of Theorem 1 of Govindan and Wilson \cite{GW2005}, which shows that a component of equilibria that is uniformily hyperstable is essential. In order to facilitate comparison when we refer to the proof of that theorem, we  have also subdivided our proof into three steps, each one corresponding to one of the steps in that proof. 

For Step 1, we have to extend the notion of equivalence between games from normal to strategic form. We say that a strategic-form game $\bar G$  with strategy space $P = \prod_n P_n$ is equivalent to the normal-form game $G$ if there exists, for each $n \in \N$, an affine surjective map $g_n: P_n \to \S_n$ such that letting $g \equiv \times_{n } g_n, \forall p \in P, \bar G_n(p) = G_n(g(p))$.

\subsection*{Step 1}  Let $\BR^G$ be the best-reply correspondence of $G$ and let $W$ be a neighborhood of the graph of $\BR^{G}$. If $\bar G$ is equivalent to $G$, then there is a corresponding neighborhood $\bar W$ of the best-reply correspondence $\BR^{\bar G}$ of $\bar G$.  Fix $\eta > 0$.  In this step, we construct a strategic-form game $\bar G$ that is equivalent to $G$ with the following properties. 

\begin{enumerate} 
	
	\item[(0)] The strategy set $\bar P_n$ of each $n$ is $\S_{n} \times \S_{n+1}$.
	\item[(1)] For each $n$, there is a simplicial complex $\mathcal{I}_n$, with $|\I_n| = \S_n$,  and a simplicial subdivision $\I_n^*$ of $\I_n$ such that:
	\begin{enumerate}
		\item the diameter of each simplex of $\I_n$ is less than $\eta$; 
		\item for each $i = 1, \ldots, k$ and $j = 1, \ldots, |c_i|$ there is a distinguished full-dimensional multisimplex $\bar K^{ij}$ of the complex $|\mathcal I^*| \times |\mathcal I^*|$ that is contained in $U_i \times U_i$.  
	\end{enumerate}
	\item[(2)] For each $n$, there exists a multisimplicial map $\bar f_n: \prod_{m \neq n} (|\I_{m}^*| \times |\I_{m+1}^*|)  \to |\I_n| \times |\I_{n+1}^*|$ such that letting $\bar f = \times_n \bar f_n$: 
	\begin{enumerate}
		\item the graph of $\bar f$ is contained in $\bar W$;
		\item the only multisimplices left fixed by $\bar f$ are the $\bar K^{ij}$'s;
		\item the restriction of $\bar f$ to each $\bar K^{ij}$ is a homeomorphism onto its image  $\bar L^{ij}$ and has a unique fixed point $\bar \s^{ij}$, whose index is the sign of $c_i$.
	\end{enumerate} 
\end{enumerate}

This step follows once we obtain an appropriate multisimplicial function $f$ from $\S$ to itself with properties that track those listed above.  The following claim constructs such a multisimplicial function.

\begin{claim}\label{claim1} There exist a multisimplicial complex $\T$, with $|\T| = \S$, and a multisimplicial subdivision $\T^*$ of $\T$  with the following properties.
\begin{enumerate}
		\item[(1a)] The diameter of each multisimplex of $\T$ is less than $\eta$; 
		\item[(1b)] For each $i = 1, \ldots, k$ and $j = 1, \ldots, |c_i|$ there is a a distinguished full-dimensional multisimplex $K^{ij}$ of $\T^*$ that is contained in $U_i$.  
	\item[(2)] There exists a multisimplicial map $f: |\T^*|  \to |\T|$ such that: 
	\begin{enumerate}
		\item the graph of $f$ is contained in $W$;
		\item the only multisimplices left fixed by $f$ are the $K^{ij}$'s;
		\item the restriction of $f$ to each $K^{ij}$ is a homeomorphism onto its image $L^{ij}$ and has a unique fixed point $\s^{ij}$, whose index is the sign of $c_i$.
	\end{enumerate} 
\end{enumerate}	
\end{claim}

\begin{proof}[Proof of Claim \ref{claim1}]
	
Approximate the correspondence $BR^{G}$ by a continuous map $h: \S \to \S$ such that: (1) $h(\S) \subset \S \backslash \partial \S$ and the graph of $h$ is contained in $W$; (2) all the fixed points of $h$ are in $\cup_i U_i$; (3) in each $U_i$,  $h$ has a unique fixed point $\s^i(h)$ with index $c_i$. For $\z >0$, define $B_{\z}(\s)$ to be the ball arround $\s$ with radius $\z$. Let $\z>0$ be such that for each $i$ and each $\s$ in $V_i \equiv B_\z(\s^i(h))$ the following hold: (A1) $\{\, \s \, \} \times B_\z(\s) \subset W$;  (A2) $B_\z(\s) \subset \S \backslash \partial \S$; (A3) the displacement of $\s$ under $h$ has norm less than $\z$.

For each $i$, pick points $\s^{ij}$, $j = 1, \ldots, |c_i|$ in $V_i$.  For each $i,j$, and $n$,  let $X_n^{ij} \subseteq Y_n^{ij}$ be full-dimensional simplices with $\s^{ij}_n$ as their barycenter and such that: for each $n$, the $Y_n^{ij}$'s are pairwise disjoint; and letting $Y^{ij} \equiv \prod_{n} Y_n^{ij}$, we have that $Y^{ij}$ is contained in $V_i$ with diameter of $Y^{ij}_n$ less than $\z$.  For each $i,j$, there exists a simplicial homeomorphism $\hat f_n^{ij}: X_n^{ij} \to Y_n^{ij}$ such that $\s^{ij}_n$ is the unique fixed point of $\hat f_n^{ij}$ and the index of $\s^{ij}$ relative to $\hat f^{ij} \equiv \times_n \hat f^{ij}_n: X^{ij} \to Y^{ij}$ is the sign of $c_i$, where  $X^{ij} = \prod_n X_n^{ij}$. The maps $\hat f^{ij}$ define a map $\hat f: \cup X^{ij} \to \cup Y^{ij}$. Let $\hat d$ be the displacement of $\hat f$ and let $d$ be the displacement of $h$.  By the Hopf Extension Theorem, the displacement $\hat d$ of $\hat f$ can be extended to the whole of $\S$ in such a way that: (B1) $\hat d(\s) = d(\s)$, if $\s \notin \cup_i V_i$; (B2) $\Vert \hat d(\s) \Vert < \z$, if $\s \in V_i$; (B3) the only zeros are the $\s^{ij}$'s. The displacement $\hat d$ now defines a map $\hat f$ from $\S$ to itself, thanks to (B1), (A2), and (A3). From (B1), (B2), and (A1), it follows that Graph$(\hat f) \subset W$.

From (B3), the map $\hat f$ has no fixed points in $\S \setminus \cup \text{int} (X^{ij})$. Let $\a>0$ be such that $||\hat f(\s) -\s|| > \a$ for $\s \notin \cup_{ij} X^{ij}$. Consider now a simplicial complex $\T_n$ whose underlying space is $\S_n$ such that: (i) $\T_n$ contains $Y^{ij}_n$ as a subcomplex for each $i,j$; (ii) $\T_n$ has diameter less than $\frac{\a}{2}$ and $\eta$; (iii) for each $i,j$ there exists a maximal simplex $L^{(0)ij}_n \in \T_n$ that contains $\s^{ij}_n$ in its interior; (iv) for each $i,j$ and any multisimplex $L$ of $Y^{ij}$ that is not an $L^{(0)ij}$, $L \cap \hat f(L) = \emptyset$.  As $\hat f_n|_{X^{ij}_n}: X^{ij}_n \to Y^{ij}_n$ is an affine simplicial homeomorphism it follows that $\mathcal{S}^{ij}_n = \{ \hat f^{-1}_n(L_n) \cap X^{ij}_n \}_{L_n \in \T_n}$ is a simplicial complex with $X^{ij}_n$ as its underlying space. Denote by $K^{(0)ij}_n \in \mathcal{S}^{ij}_n$ the simplex containing $\s^{ij}_n$ in its interior. Take now a simplicial complex $\T^*_n$ that is a subdivision of $\T_n$ and such that: $\T^*_n \cap X^{ij}_n$ is a simplicial subdivision of $\mathcal{S}^{ij}_n$;  $\s_n^{ij}$ belongs to the interior of a maximal simplex $K^{(1)ij}_n$ of $\T^*_n$. The multicomplexes $\T^*$ and $\T$ satisfy properties 1(a) and 1(b).  There remains to construct a multisimplicial approximation $f$ to $\hat f$ satisying the properties (2)(a)-(c).

Taking the diameter of $\T_n$ sufficiently small, above, guarantees that the graph of every multisimplicial approximation of $\hat f$ is contained in $W$. Fix such a diameter for $\T_n$; taking now the diameter of $\T^*_n$ sufficiently small guarantees a multisimplicial simplicial approximation of $\hat f$ indeed exists. Moreover, by properties (ii) and (iv) above, for each such approximation, a multisimplex that is not contained in $K^{(0)ij}$ for some $i,j$ will not be held fixed.  By a careful choice of the vertex map on the multisimplices of $K^{(0)ij}$ we will satisfy the remaining properties required in (2).

For each $i,j,n$, map the vertices of $K_n^{(1)ij}$ 1-1 and onto the vertices of  $L_n^{(0)ij}$ such that the resulting multisimplicial map from $K^{(1)ij}$ to $L^{(0)ij}$ has $\s_n^{ij}$ as the unique fixed point and its index is the sign of $c_i$.  For a vertex $v^*_n$ of $\T^*_n$ in $K^{(0)ij}_n$ that is not a vertex of $K^{(1)ij}_n$, define the following vertex map: consider the ray from $\s^{ij}_n$ through the vertex $v^*_n$ and let $r(v^*_n)$ be the point of intersection of the ray with the boundary of $K^{(0)ij}_n$. Let $v_n$ be closest vertex of $L^{(0)ij}_n$ to $\hat f_n(r(v^*_n))$ in the carrier of $\hat f_n(r(v^*_n))$ with respect to $\T_n$. The resulting assignment of vertices defines a simplicial approximation $f_n$ of $\hat f_n$ from $K^{(0)ij}_n$ to $L^{(0)ij}_n$ with $K^{(1)ij}_n$ as the unique fixed simplex. Extend the vertex map for $f$ to the other vertices of $\T$ to construct a multisimplicial map $f$.  The map $f$ has all the claimed properties.
\end{proof}

Let $f$ be the map obtained from Claim \ref{claim1}. We now define the following strategic-form game $\bar G$, which is equivalent to $G$.  For each $n$, the strategy set is $\bar P_n \equiv \S_n \times \S_{n+1 (\text{mod N})}$, where the second factor is payoff irrelevant. Let $\rho_{n,n}$ and $\rho_{n,n+1}$ be the projections to the first and second factor, respectively.  Define now $\bar f_n: \bar P_{-n} \to \bar P_n$ by letting $\bar f_n(p_{-n}) = f_n(\rho_{-n}(p_{-n}), \rho_{n-1,n}(p_{n-1})) \times \rho_{n+1,n+1}(p_{n+1})$.  It is clear that $\bar P$ and $\bar f$ have the properties we set out to establish. 

\subsection*{Step 2} Let now $I_n$ and $I^*_n$ denote the set of vertices of $\I_n$ and $\I^*_n$ obtained in Step 1 above. For each $n$, let $\tilde P_n = \D(I_n) \times \D(I^*_{n+1})$.  Since the vertices of $\tilde P_n$ can be viewed as points in $\bar P_n$, there is a natural map from $\tilde P_n$ to $\bar P_n$ that sends $\tilde p_n$ to the corresponding convex combination of vertices of $\I_n \times \I^*_{n+1}$.  Let $\tilde G$ be the strategic-form game where the strategy polytope of player $n$ is $\tilde P_n$ and the payoff of each player $n \in \N$ from  $\tilde s = ((i_n, i_{n+1}^*))_{n \in \N} \in \prod_{n \in \N}(I_n \times I^*_{n+1})$ is defined as $\tilde G_n(\tilde s) = \bar G_n(\tilde s)$. Perturb now $\tilde G$'s payoffs such that if $\bar p$ is a vertex of a simplex $\bar K^{ij}$ (cf. Step 1, (2.c)), all vertices of $\bar L_n^{ij}$ are equally good replies against it.  If $\eta$ is small, this perturbation of $\tilde G$ will indeed be small. We replace $\tilde G$ with this perturbed game but, for notational convenience, call it still $\tilde G$.  We show that when $W$ is a small neighborhood of $\BR^G$ and $\eta$ is also small, there exists a map $g: \tilde P \to [0, \frac{\d}{2}]^{R}$, $R \equiv \sum_{n \in \N}|I_n \times I^*_{n+1}|$, such that:

\begin{enumerate}

\item for each $n \in \N$, $g_n = g^0_n + g^1_n$, where $g^0_n$ and $g^1_n$ are continuous functions that are independent of $\tilde p_n$;

\item if $\tilde p$ projects to $\bar K^{ij}$ for some $i,j$, then for each player $n$, and every pair of  vertices $(i_n, i^*_{n+1}), (\tilde i_n, \tilde i^*_{n+1})$ of $\bar L_n^{ij}$, we have that $g_{n,(i_n,i^*_{n+1})}^0(\tilde p) = g_{n,(\tilde i_n, \tilde i^*_{n+1})}^0(\tilde p)$;

\item if $\tilde K$ is a face of $\tilde P$ whose vertices span a multisimplex of $\I \times \I^*$, then for each $(i_n, i_{n+1}^*) \in I_n \times I^*_{n+1}$, the map $g^1_{n, i_n, i_{n+1}^*}$ is multilinear in $\tilde K$;

\item $\tilde p$ is an equilibrium of the finite game $\tilde G \bigoplus g(\tilde p)$ iff it projects to $\bar \s^{ij}$ for some $i,j$ and the support of $\tilde p$ spans a multisimplex of $\I \times \I^*$; moreover if $\tilde p$ is an equilibrium of $\tilde G \bigoplus g(\tilde p)$, then it is an isolated equilibrium in this game and has the same index as its projection $\bar \s^{ij}$.

\end{enumerate}

We start by defining the map $g^0$.  The $(i_n, i_{n+1}^*)$ coordinate of $g_n^0$ is independent of $i_{n+1}^*$.  First we define the map in $\bar P$ and then extend it to $\tilde P$ by the equivalence relation between strategies in $\bar G$ and $\tilde G$.  We shall write $\bar p_n = (\bar p_{n,n}, \bar p_{n,n+1}) \in \S_n \times \S_{n+1} = \bar P_n$. Fix $n \in \N$, $\bar p_{-n} \in \bar P_{-n}$ and let $\s_{-n}$ be the strategy profile in $G$ that is equivalent to $\bar p_{-n}$. Let $g^0_{n, i_n}(\bar p_{-n}) = \pi_{i_n}(\bar p_{-n})[r_{n}(\bar p_{-n}) - G_n(i_n, \s_{-n})]$, where $r_n(\bar p_{-n}) \equiv \max_{s_n \in S_n}G_n(s_n, \s_{-n})$ and $\pi_{i_n}$ is a Urysohn function defined on $\bar P_{-n}$ that is $1$ on the inverse image under $\bar f_{n,n}$ of the closed star of $i_n$ and strictly less than 1 elsewhere.  Following the same reasoning as in Step 2 of GW, for sufficiently small $\eta>0$ and an appropriate choice of the neighborhood $W$ from which to obtain the map $\bar f$ in Step 1, one guarantees that $||g^0||_{\infty}<\frac{\d}{4}$.

We will now define $g^1$.  For each $n$, $\bar p_{-n} \in \bar P_{-n}$, and a vertex $(i_n, i_{n+1}^*)$ of $I_n \times I_{n+1}^*$, let $\bar f_{n,n}(\bar p_{-n})(i_n)$ and  $\bar f_{n,n+1}(\bar p_{-n})(i^*_{n+1})$ be, respectively,  the $i_n$-th and $i_{n+1}^*$-th barycentric coordinate of $\bar f_{n,n}(\bar p_{-n})$ and $\bar f_{n,n+1}(\bar p_{-n})$. Also, let $\s^{ij}_{i_n}$ be the $i_n$-barycentric coordinate of $\s^{ij}_n$ in $\I_n$ and $\s^{ij}_{i^*_{n+1}}$ be the $i^*_{n+1}$-barycentric coordinate of $\s^{ij}_{n+1}$ in $\I^*_{n+1}$.  For each $(i_n, i^*_{n+1}) \in I_n \times I^*_{n+1}$ that is not a vertex of $\bar K_n^{ij}$ (cf. Step 1, (2.c)), 

$$g^1_{n, i_{n}}(\bar p_{-n}) = \frac{\d}{8}(\bar f_{n,n}(\bar p_{-n})(i_n)),$$ $$g^1_{n, i^*_{n+1}}(\bar p_{-n}) = \frac{\d}{8}(\bar f_{n,n+1}(\bar p_{-n})(i^*_{n+1}));$$ 

in case $(i_n, i^*_{n+1})$ is a vertex of $\bar K_n^{ij}$, let $$g^1_{n, i_n}(\bar p_{-n}) = \frac{\l}{|S_n|\s^{ij}_{i_n}}(\bar f_{n,n}(\bar p_{-n})(i_n)),$$ $$g^1_{n, i^*_{n+1}}(\bar p_{-n}) = \frac{\b}{|S_{n+1}|\s^{ij}_{i^*_{n+1}}}(\bar f_{n,n+1}(\bar p_{-n})(i^*_{n+1})),$$ 

where $0<\l < \frac{(\d |S_n| \s^{ij_i}_{i_n})}{8}$, $0< \b < \frac{(\d |S_{n+1}| \s^{ij_i}_{i^*_{n+1}})}{8}$. Let $g^1_{i_n, i_{n+1}} = g^1_{n, i_n} + g^1_{n, i^*_{n+1}}$. It follows that $||g^1||_{\infty} \le \frac{\d}{4}$.

Points  (1) and (3) of this step are now obvious.  Point (2) holds from the definition of $g^0$ along with the fact that we have already perturbed the payoffs in $\tilde G$ accordingly. We now prove (4). 

Suppose that $\tilde p \in \tilde P$ is an equilibrium of $\tilde G \bigoplus g(\tilde p)$ and let $\bar p$ be its projection to $\bar P$.  Consider  first the case that $\bar p$ is contained in the interior of a multisimplex $\bar K$ that is different from $\bar K^{ij}$, for all $i, j$. From Step 1, it follows that $\bar K$ is not a fixed multisimplex of $\bar f$. Letting $\bar L$ be the lowest dimensional multisimplex containing $\bar f(\bar K)$, there exists a player $n$ and vertex $(i_n, i^*_{n+1})$, such that: either $\bar p_{n,n}$ has a positive $i_n$-th barycentric coordinate whereas $f_{n,n}(\bar p_{-n})$ has zero as $i_n$-th barycentric coordinate; or $\bar p_{n,n+1}$ has  a positive $i^*_{n+1}$-th barycentric coordinate whereas $f_{n,n+1}(\bar p_{-n})$ has zero as $i^*_{n+1}$-th barycentric coordinate. By the construction of the bonus vectors $g^0$ and $g^1$, this implies that either the $i_n$-th coordinate or the $i_{n+1}^*$-th coordinate is not a best reply to $\tilde p$ and thus $\tilde p$ cannot be an equilibrium in $\tilde G \bigoplus g(\tilde p)$.  

Suppose now that $\bar p$ belongs to $\bar K^{ij}$ for some $i,j$. By the definition of the bonus vector $g$, it is clear that $\bar p = \bar \s^{ij}$ and that $\tilde p$ is in fact an isolated equilibrium of the game $\tilde G \bigoplus g(\tilde p)$. To complete the proof of this step, we have to show its index is the sign of $c_i$.  By the construction of the bonus vector, the strategies that are not vertices of $\bar K_n^{ij}$ are not best replies to $\tilde p$.  Hence the index can be computed by the restriction of the best-reply correspondence of $\tilde G \bigoplus g(\tilde p)$ to $\bar K^{ij}$.  Moreover, it can be viewed as a correspondence from $\bar K^{ij}$ to $\bar L^{ij}$.  The best-reply correspondence is, then, (linearly) homotopic to $\bar f$ and has just $\bar p$ as its only fixed point along the homotopy and therefore its index is the sign of $c_i$.

\subsection*{Step 3} We now complete the proof of Lemma \ref{lemma1} by taking $\tilde G$ in Step 2 and constructing a normal-form game equivalent to $G$ with the required properties. The proof of this step is divided in two parts. First we construct an equivalent strategic-form game $\hat G$ and prove that there exists a $\d$-perturbation $\hat G^{\d}$ of $\hat G$ such that each equilibrium of $\hat G^{\d}$ is: (3.i) isolated; (3.ii) projects to $\bar \s^{ij}$, for some $i,j$; (3.iii) has the same index as its projection $\bar \s^{ij}$. The second part of the proof constructs an equivalent normal-form game $G^*$ from $\hat G$ and shows that the desired properties (1) and (2) in the statement of Lemma \ref{lemma1} are satisfied for the game $G^*$.

Consider a multisimplicial approximation $g^{0,*}: |\hat \F \times \hat \F| \to [0, \frac{\d}{4}]^{R}$ of $g^0$ (cf. Step 2, (1)), such that the following properties are satisfied: \begin{enumerate}

\item[(3.1)] for each player $n$, $\hat \F_n$ is a simplicial subdivision of $\D(I^*_n)$; 
\item[(3.2)] for each $n \in \N$, $v_n, v'_n$ vertices of $\bar L^{ij}_n$, $\tilde p \in \tilde P$ projecting to $\bar p \in \bar K^{ij}$, $g^{0,*}_{n,v_n}(\tilde p_{-n}) = g^{0,*}_{n,v'_n}(\tilde p_{-n})$; 
\item[(3.3)] Let $g^* \equiv g^{0,*} + g^1$. The profile $\tilde p$ is an equilibrium of  $\tilde G \bigoplus g^*(\tilde p)$ iff it projects to $\bar \s^{ij}$ for some $i,j$ and the support of $\tilde p$ spans a multisimplex of $\I \times \I^*$. Moreover, the index of an equilibrium $\tilde p$ computed from the best-reply $\BR^{g^*}$ of $\tilde G \bigoplus g^*(\tilde p)$ is equal to the index of $\bar \s^{ij}$.
\end{enumerate}

To show the existence of such a multisimplicial approximation, argue as follows.  First, property (3.1) is immediate from the definition of a multisimplicial approximation. Second, since $g^0$ satisfies property (2) of Step 2, we can take the multisimplicial approximation to satisfy (3.2). Third,  property (3.3) is now a consequence of property $(3.2)$ if the approximation is sufficiently fine.  Indeed, thanks to Property (4) in Step 2, taking the approximation $g^{0,*}$ to be sufficiently fine, we have that an equilibrium $\tilde p$ of $\tilde G \bigoplus g^*(\tilde p)$ must project to $\bar K^{ij}$ for some $i,j$. Because of (3.2) and the definition of $g^*$, it follows that the best-replies of players in $\tilde G \bigoplus g^*(\tilde p)$ and $\tilde G \bigoplus g(\tilde p)$ are identical when $\tilde p$ projects to $\bar K^{ij}$, for some $i,j$. This implies, by (4) in Step 2, that $\tilde p$ is an equilibrium of $\tilde G \bigoplus g^*(\tilde p)$ if and only if $\tilde p$ projects to $\bar \s^{ij}$ and the support of $\tilde p$ spans a multisimplex of $\I \times \I^*$.  And as there, it is easy to see that the index has the requisite sign. 

By Lemma 2.3.15 in \cite{LRS2010}, there exists a subdivision $\hat \F_n^*$ of  $\hat \F_n$ and a piecewise-linear and convex function $\g_n: |\hat \F_n| \to \Re_{+}$  that is linear precisely on the simplices of $\hat \F_n^*$ (cf.  also GW, Appendix B). Given the degrees of freedom in the choice of $\g_n$, we will also assume for each $i,j$, the following property holds for the simplex $\hat K_n^{ij} \in \hat \F^*_n$ that contains in its interior the unique point $\tilde p_{n,n}$ that projects to $\s_n^{ij}$:

\begin{enumerate}
\item [(3.4)] for  any pair of vertices $v^{ij}_n, v^{'ij}_n$ of $\hat K^{ij}_n$, $\g_n(\hat v^{ij}_{n}) = \g_n(\hat v'^{ij}_{n})$.
\end{enumerate}

We now define the strategic-form game $\hat G$. Let $\hat F_n^*$ be the set of vertices of $\hat \F_n^*$. The strategy set of player $n$ is $\hat P_n = \D(\hat F_n^*) \times \D(\hat F_{n+1}^*)$ and the payoff of a vertex $\hat v$ of $\hat P$ to player $n$ is $\hat G_n(\hat v) = \tilde G_n(\hat v)$. The game $\hat G$ is clearly equivalent to $\tilde G$. We now define the $\d$-perturbation of this game satisfying ($3$.i-iii).

Let $A_{n,n}$ be the $(|I_{n}| \times |\hat F_n^*|$)-matrix such that the $\hat v_n$-column $i_{n}$-th row is defined by the $i_n$-barycentric coordinate of the vertex $\hat v_{n}$ in the simplicial complex $\I _n$. Let $A_{n,n+1}$ be the $(|I^*_{n+1}| \times |\hat F_{n+1}^*|)$-matrix defined similarly.  Let $A'_{n,s}$ be the transpose of $A_{n,s}$, where $s=n, n+1$. Let $g^1_{n,n} \equiv \times_{i_n}g^1_{n, i_n}$ and $g^1_{n,n+1} \equiv \times_{i^*_{n+1}}g^1_{n, i^*_{n+1}}$ Define then $\hat g^{1,*}_{n,n} = A'_{n,n} g^1_{n,n}: \hat P_{-n} \to [0, \frac{\d}{4}]^{|\hat F_{n}^*|}$ and $\hat g^{1,*}_{n,n+1} = A'_{n,n+1} g^1_{n,n+1}: \hat P_{-n} \to [0, \frac{\d}{4}]^{|\hat F_{n+1}^*|}$ (where $g^{1,*}_{n,s}$ is defined first over $\tilde P$ and then, from the equivalence relation between strategies in $\tilde P$ and $\hat P$, $g^{1,*}_{n,s}$ is defined on $\hat P$).  Similarly,  let $ \hat g^{0,*}_n$ denote $A'_{n,n} g^{0,*}_n$.  For each $\a > 0$, define now the  game $\hat G^{\a}$ of $\hat G$ as follows: for each $n \in \N$, the payoff to player $n$ from a vertex $\hat v \in \hat P$ is: $\hat G_n(\hat v) + \hat g^{0,*}_{\hat v_{n,n}}(\hat v_{-n}) + \hat g^{1,*}_{n,\hat v_{n,n}}(\hat v_{-n}) - \alpha \g_{n}(\hat v_{n,n}) + \hat g^{1,*}_{n,\hat v_{n,n+1}}(\hat v_{-n}) - \alpha \g_{n+1}(\hat v_{n,n+1})$, where $\hat v_{n,n}$ is a vertex of player $n$ in $\hat F_n^*$ and $\hat v_{n,n+1}$ is a vertex of player $n$ in $\hat F_{n+1}^*$; $\hat v_{n}$ denotes the pair $(\hat v_{n,n}, \hat v_{n,n+1})$ and $\hat v_{-n} = (\hat v_{j})_{j \neq n}$.

We will now prove that $\hat G^\a$ satisfies properties (3.i) to (3.iii), if $\a$ is sufficiently small.  Take an equilibrium $\hat p$ of $\hat G^{\a}$. For $\hat p_n \in \hat P_n$, denote by $\hat p_{n,n}$ (resp. $\hat p_{n,n+1}$) the coordinates of $\hat p_n$ in $\D(\hat F_n^*)$ (resp. $\D(\hat F_{n+1}^*)$).  It follows from the construction of $\g_n$ (resp. $\g_{n+1}$) that the support of $\hat p_{n,n}$ (resp. $\hat p_{n,n+1}$) must span a simplex in $\hat \F_n^*$ (resp. $\hat \F_{n+1}^*$). Note now that each coordinate map of $g^{0,*}$ is multilinear in each face of $\hat P$ that projects to a multisimplex of $\hat \F^*$; $g^{1,*}$ is also multilinear in each such face of $\hat P$, since each of its coordinate maps is multilinear in a face of $\tilde P$ (cf. (3) of Step 2) that projects to a multisimplex of $\I^* \times \I^*$ (which $\hat \F^* \times \hat \F^*$ subdivides) . Hence the payoff to player $n$ of a vertex $\hat v_n \in \hat P_n$ against $\hat p_{-n}$ in the game $\hat G^{\a}$ is: $\hat G_n(\hat v_n, \hat p_{-n}) + g^{0,*}_{\hat v_{n,n}}(\hat p_{-n}) + g^{1,*}_{n,\hat v_{n,n}}(\hat p_{-n}) - \alpha \g_{n}(\hat p_{n,n}) + g^{1,*}_{n,\hat v_{n,n+1}}(\hat p_{-n}) - \alpha \g_{n+1}(\hat p_{n,n+1})$. We claim that for $\a>0$ sufficiently small, $\hat p$ must project to $\bar K^{ij}$. Aiming at a contradiction, let $\alpha_k \to 0$ and $(\hat p^k)_{k \in \mathbb{N}}$ a sequence of equilibria of $\hat G^{\a_k}$ such that $\hat p^k$ projects to the complement of  $\cup \, \text{int}(\bar K^{ij})$. Assume without loss of generality that the sequence converges to some $\hat p$, by passing to a convergent subsequence if necessary. Then $\hat p$ is an equilibrium of $\hat G^{0}$. Let $\tilde p \in \tilde P$ be such that $\hat p$ projects to $\tilde p$. Then $\tilde p$ projects to the complement of $\cup \, \text{int}(\bar K^{ij})$ and is an equilibrium of $\tilde G \bigoplus g^*(\tilde p)$, which contradicts  (3.3).  Thus for all small $\a$, there is no equilibrium of $\hat G^\a$ that projects to a point outside $\bar K^{ij}$ for each $i,j$.

Let now $\hat K^{ij}$ be the multsimplex of $\hat \F^* \times \hat \F^*$ whose projection to $\bar P$ is in $\bar K^{ij}$ and contains $\bar \s^{ij}$ in its interior. Taking $\a >0$ sufficiently small, as we saw above, no equilibrium of $\hat G^{\a}$ projects to the complement of $\cup \hat K^{ij}$. Now, (3.4) and (3.3) imply that the unique equilibrium $\hat p$ in $\hat K^{ij}$ must project to $\bar \s^{ij}$. This proves that $\hat G^{\a}$ satisfies (3.i) and (3.ii).

We now prove (3.iii). Fix $i, j$ and let $\hat p^{ij} \in \hat P$ project to $\bar \s^{ij}$.  There is a component $\hat C^{ij}$ of equilibria of the game $\hat G \bigoplus \hat g(\hat p^{ij})$ that includes $\hat p^{ij}$.  As the index is invariant to the addition/deletion of equivalent strategies, the index of $\hat C^{ij}$ is the sign of the equilibrium $\tilde p^{ij}$ in the game $\tilde G \bigoplus g^*(\tilde p^{ij})$, where $\tilde p^{ij}$ is the point in $\tilde P$ projecting to $\bar \s^{ij}$, and thus the index of $\hat C^{ij}$ is the sign of $c_i$.  For small $\a$, the unique equilibrium close to this component $\hat C^{ij}$ is in fact $\hat p^{ij}$. Thereore $\hat p^{ij}$ retains the index of $\hat C^{ij}$, giving us property (3.iii). 

We now proceed to the proof of the second part, claimed at the begining of Step 3, namely:  there exists an equivalent normal-form game $G^*$ to $\hat G$ (obtained by adding duplicates to $\hat G$) and a $\d$-perturbation of this game such that properties (1) and (2) in the statement of Lemma \ref{lemma1} are satisfied.

Consider a simplicial complex $\H^*_n$ whose space is  $\hat P_n$ for which there exists  a convex piecewise linear function $q_n: |\H^*_n| \to \Re_{+}$ such that: $q_n$ is linear precisely on the simplices of $\H^*_n$; and if $K^*_n$ is the simplex of $\H^*_n$ containing $\hat p^{ij}$ in its interior, then $q_n$ is constant in each vertex of this simplex. Let $H^*_n$ be the set of vertices of $\H^*_n$. Define now a normal-form game $G^*$, for which the set of pure strategies of each player $n$ is $H^*_n$. Define the payoffs for each $n \in \N, h^* \in H^* \equiv \prod_{n \in \N}H^*_n$ by $G^*_n(h^*) \equiv \hat G_n(h^*)$. The game $G^*$ is obviously equivalent to $G$, since $\hat G$ is.  Define now the following perturbation $G^{*,\xi}$ of $G^*$: for each $n \in \N, h^* \in H^*$, let $G^{*,\xi}_n(h^*) = \hat G^{\a}_n(h^*) - \xi q_n(h^*_n)$. 

Let $\xi>0$. Let the mixed strategy set of player $n$ in the game $G^*$ be denoted by $\S^*_n$. Given $\s^* \in \S^*$, the pure best-replies of player $n$ in $G^{*,\xi}$ to $\s^*_{-n}$ are the vertices of a simplex of $\H^*_n$, due to the fact that $q_n$ is convex and piecewise linear. Therefore, an equilibrium $\s^*$ of $G^{*,\xi}$ must be such that for each $n \in \N$, $\s^*_n$ has support in the vertices of a simplex of $\H^*_n$. Taking $\xi>0$ sufficiently small, it then follows that $\s^*$ is an equilibrium of $G^{*,\xi}$ if and only if $\s^*$ is isolated and projects to an equilibrium of $\hat G^{\a}$. Therefore, we have that (1) of Lemma \ref{lemma1} is satisfied. The proof of (2) follows the same logic we employed to get the index of the equilibria of $\hat G^\a$.

\end{document}